\documentclass[reqno]{amsart}							
\usepackage{microtype}

\usepackage{amsmath}								
\usepackage{amssymb}								
\usepackage{amsfonts}								
\usepackage{amsthm}									
\usepackage[foot]{amsaddr}								

\usepackage{mathtools}								

\usepackage[utf8]{inputenc}								
\usepackage[T1]{fontenc}								

\usepackage[sf,mono=false]{libertine}

\usepackage[scale=0.95,ttdefault=true]{AnonymousPro}

\usepackage{dsfont}									

\usepackage[
cal=cm,
]
{mathalfa}

\usepackage[labelfont={bf,small},labelsep=colon,font=small]{caption}		
\usepackage[dvipsnames,svgnames]{xcolor}						
\colorlet{MyBlue}{DodgerBlue!75!Black}
\colorlet{MyGreen}{DarkGreen!85!Black}

\newcommand{\afterhead}{.}								

\usepackage{subfigure}								
\usepackage{tikz}									
\usetikzlibrary{calc,cd,patterns}

\usepackage{acronym}								
\usepackage{booktabs}								
\usepackage{latexsym}								
\usepackage{paralist}									
\usepackage{xspace}									

\usepackage[authoryear,sort&compress]{natbib}					

\newcommand{\citepos}[1]{\citeauthor{#1}'s \textpar{\citeyear{#1}}}

\usepackage{hyperref}
\hypersetup{
colorlinks=true,
linktocpage=true,
pdfstartview=FitH,
breaklinks=true,
pdfpagemode=UseNone,
pageanchor=true,
pdfpagemode=UseOutlines,
plainpages=false,
bookmarksnumbered,
bookmarksopen=false,
bookmarksopenlevel=1,
hypertexnames=true,
pdfhighlight=/O,
urlcolor=Maroon,linkcolor=MyBlue!60!black,citecolor=DarkGreen!70!black,		
pdftitle={},
pdfauthor={},
pdfsubject={},
pdfkeywords={},
pdfcreator={pdfLaTeX},
pdfproducer={LaTeX with hyperref}
}

\newcommand{\EMAIL}[1]{\email{\href{mailto:#1}{#1}}}

\numberwithin{equation}{section}							
\usepackage[sort&compress,capitalize,nameinlink]{cleveref}				
\crefname{example}{Example}{Examples}
\crefrangeformat{equation}{\upshape(#3#1#4)\textendash(#5#2#6)}

\DeclarePairedDelimiter{\bracks}{[}{]}						
\DeclarePairedDelimiter{\parens}{(}{)}						

\DeclarePairedDelimiter{\abs}{\lvert}{\rvert}					
\DeclarePairedDelimiter{\norm}{\lVert}{\rVert}					
\DeclarePairedDelimiterXPP{\dnorm}[1]{}{\lVert}{\rVert}{_{\ast}}{#1}		

\DeclarePairedDelimiterX{\braket}[2]{\langle}{\rangle}{#1,#2}			
\DeclarePairedDelimiterX{\product}[2]{\langle}{\rangle}{#1,#2}			

\DeclarePairedDelimiterX{\setdef}[2]{\{}{\}}{#1:#2}				
\DeclarePairedDelimiterXPP{\exclude}[1]{\mathopen{}\setminus}{\{}{\}}{}{#1}

\newcommand{\eg}{e.g.,\xspace}							
\newcommand{\ie}{i.e.,\xspace}							
\newcommand{\textpar}[1]{\textup(#1\textup)}					

\newcommand{\txs}{\textstyle}							

\newcommand{\alt}[1]{#1'}								
\newcommand{\dual}[1]{#1^{\ast}}							
\newcommand{\est}[1]{\hat #1}							
\newcommand{\sol}[1][\debug x]{#1^{\ast}}						
\newcommand{\var}[1]{\tilde#1}							

\DeclareMathOperator{\bigoh}{\mathcal O}						

\newcommand{\R}{\mathbb{R}}							

\newcommand{\vdim}{\debug d}								
\newcommand{\vecspace}{\mathcal{\debug V}}						
\newcommand{\dspace}{\mathcal{\debug Y}}						
\newcommand{\unitvec}{\debug z}								

\newcommand{\state}{\debug x}
\newcommand{\dstate}{\debug y}

\DeclareMathOperator{\diam}{diam}						
\DeclareMathOperator{\vol}{vol}							

\newcommand{\base}{\debug p}								

\newcommand{\ball}{\mathbb{\debug B}}							
\newcommand{\sphere}{\mathbb{\debug S}}						
\newcommand{\radius}{\debug r}								

\newcommand{\cpt}{\mathcal{C}}							

\DeclareMathOperator{\tcone}{TC}							

\newcommand{\subd}{\partial}							

\newcommand{\feas}{\mathcal{\debug X}}							
\newcommand{\intfeas}{\feas^{\circ}}						
\newcommand{\Lip}{L}								
\newcommand{\obj}{f}								

\newcommand{\play}{\debug i}								
\newcommand{\playalt}{\debug j}								
\newcommand{\nPlayers}{\debug N}								
\newcommand{\players}{\mathcal{\debug \nPlayers}}					

\newcommand{\pay}{\debug u}								
\newcommand{\payv}{\debug v}								
\newcommand{\pot}{\debug f}								

\newcommand{\eq}{\sol}								
\newcommand{\eqset}{\debug \eq[\feas]}							
\newcommand{\game}{\mathcal{\debug G}}						

\DeclareMathOperator{\Eucl}{\Pi}							

\newcommand{\breg}{\debug D}								
\newcommand{\hreg}{\debug h}								
\newcommand{\mirror}{\debug Q}								
\newcommand{\prox}{\debug P}								
\newcommand{\strong}{\debug K}								

\newcommand{\new}[1]{#1^{+}}							
\newcommand{\act}{\debug X}								
\newcommand{\pert}{\debug w}								

\newcommand{\step}{\debug \gamma}							
\newcommand{\mix}{\debug \delta}								

\newcommand{\unitvar}{\debug Z}
\newcommand{\pertvar}{\debug W}								

\DeclareMathOperator{\ex}{\mathbb{E}}						
\DeclareMathOperator{\prob}{\mathbb{P}}					

\newcommand{\as}{\textup(a.s.\textup)\xspace}					
\newcommand{\dkl}{\debug D_{\mathrm{KL}}}						
\newcommand{\filter}{\mathcal{\debug F}}							
\newcommand{\noise}{\debug U}								
\newcommand{\snoise}{\debug \xi}								
\newcommand{\noisedev}{\debug \sigma}							
\newcommand{\noisevar}{\noisedev^{2}}						
\newcommand{\bias}{\debug b}								
\newcommand{\sbias}{\debug r}							

\providecommand\given{}								

\DeclarePairedDelimiterXPP{\exof}[1]{\ex}{[}{]}{}{
\renewcommand\given{\nonscript\,\delimsize\vert\nonscript\,\mathopen{}} #1}

\DeclarePairedDelimiterXPP{\probof}[1]{\prob}{(}{)}{}{
\renewcommand\given{\nonscript\,\delimsize\vert\nonscript\,\mathopen{}} #1}

\DeclareMathOperator*{\argmax}{arg\,max}					
\DeclareMathOperator*{\argmin}{arg\,min}						

\DeclareMathOperator{\dom}{dom}							
\DeclareMathOperator{\one}{\mathds{1}}						

\newcommand{\from}{\colon}							

\newcommand{\start}{\debug 1}								
\newcommand{\running}{\debug 1,2,\dotsc}							
\newcommand{\run}{\debug n}								
\newcommand{\runalt}{\debug k}								

\newcommand{\dd}{\:d}								
\DeclareMathOperator{\del}{\debug \nabla\!}							
\newcommand{\eps}{\varepsilon}							
\newcommand{\pd}{\partial}							

\newcommand{\insum}{\sum\nolimits}						
\newcommand{\inprod}{\prod\nolimits}						

\usepackage{algorithm}								
\usepackage{algpseudocode}								

\theoremstyle{plain}
\newtheorem{theorem}{Theorem}							
\newtheorem*{corollary*}{Corollary}							
\newtheorem{lemma}[theorem]{Lemma}						
\newtheorem{proposition}[theorem]{Proposition}					

\theoremstyle{definition}
\newtheorem*{definition*}{Definition}						
\newtheorem*{assumption*}{Assumptions}						
\newtheorem{example}{Example}							
\newtheorem*{example*}{Example}							

\theoremstyle{remark}
\newtheorem*{remark*}{Remark}							

\numberwithin{theorem}{section}							
\numberwithin{example}{section}							

\newcommand{\resource}{\debug s}								
\newcommand{\nResources}{\debug S}							
\newcommand{\resources}{\mathcal{\nResources}}					

\newcommand{\vbound}{\debug V}
\newcommand{\pexp}{\debug p}
\newcommand{\qexp}{\debug q}
\newcommand{\str}{\debug \beta}
\newcommand{\seq}{\debug a}

\newcommand{\debug}[1]{#1}							
\newcommand{\revise}[1]{#1}					

\newcommand{\farbound}{\debug c}								
\newcommand{\bbound}{\debug B}								
\newcommand{\cbound}{\debug C}								
\newcommand{\pbound}{\debug P}
\newcommand{\qbound}{\debug Q}
\newcommand{\rbound}{\debug R}

\begin{document}


\title{Bandit Learning in Concave $\nPlayers$-Person Games}

\author
[M.~Bravo]
{Mario Bravo$^{\sharp}$}
\address{$^{\sharp}$ Universidad de Santiago de Chile, Departamento de Matemática y Ciencia de la Computación}
\EMAIL{mario.bravo.g@usach.cl}

\author
[D.~S.~Leslie]
{David S. Leslie$^{\ddag}$}
\address{$\ddag$ Lancaster University \& PROWLER.io}
\EMAIL{d.leslie@lancaster.ac.uk}

\author
[P.~Mertikopoulos]
{Panayotis Mertikopoulos$^{\ast}$}
\address{$\ast$ Univ. Grenoble Alpes, CNRS, Inria, LIG 38000 Grenoble, France}
\EMAIL{panayotis.mertikopoulos@imag.fr}


\thanks{
%
%
M.~Bravo gratefully acknowledges the support provided by FONDECYT grant 11151003.
P.~Mertikopoulos was partially supported by
the Huawei HIRP flagship grant ULTRON,
and
the French National Research Agency (ANR) grant
ORACLESS (ANR\textendash 16\textendash CE33\textendash 0004\textendash 01).
Part of this work was carried out with financial support by the ECOS project C15E03.}

\newcommand{\acdef}[1]{\textit{\acl{#1}} \textup{(\acs{#1})}\acused{#1}}		
\newcommand{\acdefp}[1]{\emph{\aclp{#1}} \textup(\acsp{#1}\textup)\acused{#1}}	

\newacro{KW}{Kiefer\textendash Wolfowitz}
\newacro{SFP}{stochastic fictitious play}
\newacro{CCE}{coarse correlated equilibrium}
\newacroplural{CCE}[CCE]{coarse correlated equilibria}
\newacro{OGD}{online gradient descent}
\newacro{APT}{asymptotic pseudotrajectory}
\newacroplural{APT}{asymptotic pseudotrajectories}
\newacro{MW}{multiplicative weights}
\newacro{KL}{Kullback\textendash Leibler}
\newacro{DGF}{distance-generating function}
\newacro{MP}{mirror-prox}
\newacro{SPSA}{simultaneous perturbation stochastic approximation}
\newacro{MDS}{martingale difference sequence}
\newacro{FTRL}{``following the regularized leader''}
\newacro{OMD}{online mirror descent}
\newacro{DSC}{diagonal strict concavity}
\newacro{GD}{gradient descent}
\newacro{MIMO}{multiple-input and multiple-output}
\newacro{MD}{mirror descent}
\newacro{DA}{dual averaging}
\newacro{LHS}{left-hand side}
\newacro{RHS}{right-hand side}
\newacro{iid}[i.i.d.]{independent and identically distributed}
\newacroplural{iid}[i.i.d.]{independent and identically distributed}
\newacro{NE}{Nash equilibrium}
\newacroplural{NE}[NE]{Nash equilibria}
\newacro{VI}{variational inequality}
\newacroplural{VI}{variational inequalities}

\begin{abstract}
%
%
This paper examines the long-run behavior of learning with bandit feedback in non-cooperative concave games.
The bandit framework accounts for extremely low-information environments where the agents may not even know they are playing a game;
as such, the agents' most sensible choice in this setting would be to employ a no-regret learning algorithm.
In general, this does not mean that the players' behavior stabilizes in the long run:
no-regret learning may lead to cycles, even with perfect gradient information.
However, if a standard monotonicity condition is satisfied, our analysis shows that no-regret learning based on \acl{MD} with bandit feedback converges to \acl{NE} with probability $1$.
We also derive an
upper bound for the convergence rate of the process
that nearly matches the best attainable rate for \emph{single-agent} bandit stochastic optimization.
\end{abstract}
\maketitle

\acresetall										

\section{Introduction}
\label{sec:introduction}

The bane of decision-making in an unknown environment is \emph{regret:}
noone wants to realize in hindsight that the decision policy they employed was strictly inferior to a plain policy prescribing the same action throughout.
For obvious reasons, this issue becomes considerably more intricate when the decision-maker is subject to situational uncertainty and the ``fog of war'':
when the only information at the optimizer's disposal is the reward obtained from a given action (the so-called ``bandit'' framework), is it even possible to design a no-regret policy?
Especially in the context of online convex optimization (repeated decision problems with continuous action sets and convex costs), this problem becomes even more challenging because the decision-maker typically needs to infer gradient information from the observation of a single scalar.
Nonetheless, despite this extra degree of difficulty, this question has been shown to admit a positive answer:
regret minimization \emph{is} possible, even with bandit feedback \citep{Kle04,FKM05}.

In this paper, we consider a multi-agent extension of this framework where, at each stage $\run=\running$, of a repeated decision process, the reward of an agent is determined by the actions of all agents via a fixed mechanism:
\emph{a non-cooperative $\nPlayers$-person game.}
In general, the agents \textendash\ or players \textendash\ might be completely oblivious to this mechanism, perhaps even ignoring its existence:
for instance, when choosing how much to bid for a good in an online auction, an agent is typically unaware of who the other bidders are, what are their specific valuations, etc.
Hence, lacking any knowledge about the game, it is only natural to assume that agents will at least seek to achieve a minimal worst-case guarantee and minimize their regret.
As a result, a fundamental question that arises is
\begin{inparaenum}
[\itshape a\upshape)]
\item
whether the agents' sequence of actions stabilizes to a rationally admissible state under no-regret learning;
and
\item
if it does, whether convergence is affected by the information available to the agents.
\end{inparaenum}

\subsection*{Related work\afterhead}

In finite games, no-regret learning guarantees that the players' time-averaged, empirical frequency of play converges to the game's set of \acp{CCE}, and the rate of this convergence is $\bigoh(1/\run)$ for $(\lambda,\mu)$-smooth games \citep{SALS15,FLST16}.
In general however, this set might contain highly subpar, rationally inadmissible strategies:
for instance, \cite{VZ13} provide examples of \ac{CCE} that assign positive selection probability \emph{only} to strictly dominated strategies.
In the class of potential games, \cite{CHM17-NIPS} recently showed that the \emph{actual} sequence of play (\ie the sequence of actions that determine the agents' rewards at each stage) converges under no-regret learning, even with bandit feedback.
Outside this class however, the players' chosen actions may cycle in perpetuity, even in simple, two-player zero-sum games with full information \citep{MPP18,MLZF+18};
in fact, depending on the parameters of the players' learning process, agents could even exhibit a fully unpredictable, aperiodic and chaotic behavior \citep{PPP17}.
As such, without further assumptions in place, no-regret learning in a multi-agent setting does not necessarily imply convergence to a unilaterally stable, equilibrium state.

In the broader context of games with continuous action sets (the focal point of this paper), the long-run behavior of no-regret learning is significantly more challenging to analyze.
In the case of mixed-strategy learning, \cite{PL14} and \cite{PML17} showed that mixed-stratgy learning based on \acl{SFP} converges to an $\eps$-perturbed \acl{NE} in potential games (but may lead to as much as $\bigoh(\eps\run)$ regret in the process).
More relevant for our purposes is the analysis of \cite{Nes09} who showed that the time-averaged sequence of play induced by a no-regret \ac{DA} process with noisy gradient feedback converges to \acl{NE} in monotone games (a class which, in turn, contains all concave potential games).

\revise{The closest antecedent to our approach is the recent work of \cite{MZ18} who showed that the \emph{actual} sequence of play generated by \acl{DA} converges to \acl{NE} in the class of variationally stable games (which includes all monotone games).
To do so, the authors
first showed that a naturally associated continuous-time dynamical system converges, and then used the so-called \acdef{APT} framework of \cite{Ben99} to translate this result to discrete time.
Similar \ac{APT} techniques were also used in a very recent preprint by \cite{BBF18} to establish the convergence of a \emph{payoff-based} learning algorithm in two classes of one-dimensional concave games:
games with strategic complements,
and ordinal potential games with isolated equilibria.
The algorithm of \cite{BBF18} can be seen as a special case of \acl{MD} coupled with a two-point gradient estimation process, suggesting several interesting links with our paper.}

\subsection*{Our contributions\afterhead}

In this paper, we drop all feedback assumptions and we focus on the \emph{bandit} framework where the only information at the players' disposal is the payoffs they receive at each stage.
As we discussed above, this lack of information complicates matters considerably because players must now estimate their payoff gradients from their observed rewards.
What makes matters even worse is that an agent may introduce a significant bias in the (concurrent) estimation process of another,
{so traditional, multiple-point estimation techniques for derivative-free optimization cannot be applied (at least, not without significant communication overhead between players).}

To do away with player coordination requirements, we focus on learning processes which could be sensibly deployed in a single-agent setting and we show that, in monotone games, the sequence of play induced by a wide class of no-regret learning policies converges to \acl{NE} with probability $1$.
Furthermore, by specializing to the class of strongly monotone games, we show that the rate of convergence is $\bigoh(\run^{-1/3})$, \ie it is nearly optimal with respect to the attainable $\bigoh(\run^{-1/2})$ rate for bandit, \emph{single-agent} stochastic optimization with strongly convex and smooth objectives \citep{ADX10,Sha13}.

\revise{We are not aware of a similar \acl{NE} convergence result for concave games with general convex action spaces and \emph{bandit} feedback:
the analysis of \cite{MZ18} requires first-order feedback, while the analysis of \cite{BBF18} only applies to one-dimensional games.}
We find this outcome particularly appealing for practical applications of game theory (\eg in network routing) because it shows that in a wide class of (possibly very complicated) nonlinear games, the \acl{NE} prediction does not require full rationality, common knowledge of rationality, 
flawless execution, or even the knowledge that a game is being played:
a commonly-used, individual no-regret algorithm suffices.

\section{Problem setup and preliminaries}
\label{sec:prelims}

\subsection{Concave games\afterhead}

Throughout this paper, we will focus on games with a finite number of players $\play\in\players = \{1,\dotsc,\nPlayers\}$ and continuous action sets.
During play, every player $\play\in\players$ selects an \emph{action} $\state_{\play}$ from a compact convex subset $\feas_{\play}$ of a $\vdim_{\play}$-dimensional normed space $\vecspace_{\play}$;
subsequently, based on each player's individual objective and the \emph{action profile} $\state = (\state_{\play};\state_{-\play}) \equiv (\state_{1},\dotsc,\state_{\nPlayers})$ of all players' actions, every player receives a \emph{reward}, and the process repeats.
In more detail, writing $\feas \equiv \prod_{\play} \feas_{i}$ for the game's \emph{action space}, we assume that each player's reward is determined by an associated \emph{payoff} (or \emph{utility}) \emph{function} $\pay_{\play}\from\feas\to\R$.
Since players are not assumed to ``know the game'' (or even that they are involved in one) these payoff functions might be a priori unknown, especially with respect to the dependence on the actions of other players.
Our only structural assumption for $\pay_{\play}$ will be that
$\pay_{\play}(\state_{\play};\state_{-\play})$ is concave in $\state_{\play}$ for all $\state_{-\play}\in\feas_{-\play} \equiv \inprod_{\playalt\neq\play} \feas_{\playalt}$, $\play\in\players$.

With all this in hand, a \emph{concave game} will be a tuple $\game \equiv \game(\players,\feas,\pay)$ with players, action spaces and payoffs defined as above.
Below, we briefly discuss some examples thereof:

\begin{example}[Cournot competition]
\label{ex:game-Cournot}
In the standard Cournot oligopoly model, there is a finite set of \emph{firms} indexed by $\play=1,\dotsc,\nPlayers$, each supplying the market with a quantity $\state_{\play}\in[0,C_{\play}]$ of some good (or service), up to the firm's production capacity $C_{\play}$.
By the law of supply and demand, the good is priced as a decreasing function
$P(\state_{\mathrm{tot}})$ of the total amount $\state_{\mathrm{tot}} = \sum_{\play=1}^{\nPlayers} \state_{\play}$ supplied to the market, typically following a linear model of the form $P(\state_{\mathrm{tot}}) = a - b\state_{\mathrm{tot}}$ for positive constants $a,b>0$.
The utility of firm $\play$ is then given by
\begin{equation}
\label{eq:pay-Cournot}
\pay_{\play}(\state_{\play};\state_{-\play})
	= \state_{\play} P(\state_{\mathrm{tot}}) - c_{\play} \state_{\play},
\end{equation}
\ie it comprises the total revenue from producing $\state_{\play}$ units of the good in question minus the associated production cost (in the above,  $c_{\play}>0$ represents the marginal production cost of firm $\play$).
\end{example}

\begin{example}[Resource allocation auctions]
\label{ex:game-auction}
Consider a service provider with a number of splittable \emph{resources} $\resource \in \resources = \{1,\dotsc,\nResources\}$ (bandwidth, server time, GPU cores, etc.).
These resources can be leased to a set of $\nPlayers$ bidders (players) who can place monetary bids $\state_{\play\resource}\geq0$ for the utilization of each resource $\resource\in\resources$ up to each player's total budget $b_{\play}$, \ie $\sum_{\resource\in\resources} \state_{\play\resource} \leq b_{\play}$.
Once all bids are in, resources are allocated proportionally to each player's bid, \ie the $\play$-th player gets
\(
\rho_{\play\resource}
	= \parens{q_{\resource} \state_{\play\resource}} \big/ \parens{c_{\resource} + \sum_{\playalt\in\players} \state_{\playalt\resource}}
\)
units of the $\resource$-th resource (where $q_{\resource}$ denotes the available units of said resource and $c_{\resource}\geq0$ is the ``entry barrier'' for bidding on it).
A simple model for the utility of player $\play$ is then given by
\begin{equation}
\label{eq:pay-auction}
\pay_{\play}(\state_{\play};\state_{-\play})
	= \sum_{\resource\in\resources} \bracks{g_{\play} \rho_{\play\resource} - \state_{\play\resource}},
\end{equation}
with
$g_{\play}$ denoting the marginal gain of player $\play$ from acquiring a unit slice of resources.
\end{example}

\subsection{\acl{NE} and monotone games\afterhead}

The most widely used solution concept for non-cooperative games is that of a \acdef{NE}, defined here as any action profile $\eq\in\feas$ that is resilient to unilateral deviations, viz.
\begin{equation}
\label{eq:NE}
\tag{NE}
\pay_{\play}(\eq_{\play};\eq_{-\play})
	\geq \pay_{\play}(\state_{\play};\eq_{-\play})
	\quad
	\text{for all $\state_{\play}\in\feas_{\play}$, $\play\in\players$}.
\end{equation}
By the classical existence theorem of \cite{Deb52}, every concave game admits a \acl{NE}.
Moreover, thanks to the individual concavity of the game's payoff functions, \aclp{NE} can also be characterized via the first-order optimality condition
\begin{equation}
\label{eq:NE-var}
\braket{\payv_{\play}(\eq)}{\state_{\play} - \eq_{\play}}
	\leq 0
	\quad
	\text{for all $\state_{\play}\in\feas_{\play}$},
\end{equation}
where $\payv_{\play}(\state)$ denotes the individual payoff gradient of the $\play$-th player,
\ie
\begin{equation}
\label{eq:payv}
\payv_{\play}(\state)
	= \del_{\play} \pay_{\play}(\state_{\play};\state_{-\play}),
\end{equation}
with $\del_{\play}$ denoting differentiation with respect to $\state_{\play}$.%
\footnote{We adopt here the standard convention of treating $\payv_{\play}(\state)$ as an element of the dual space $\dspace_{\play} \equiv \dual\vecspace_{\play}$ of $\vecspace_{\play}$, with $\braket{\dstate_{\play}}{\state_{\play}}$ denoting the duality pairing between $\dstate_{\play}\in\dspace_{\play}$ and $\state_{\play}\in\feas_{\play}\subseteq\vecspace_{\play}$.}
{In terms of regularity, it will be convenient to assume that each $\payv_{\play}$ is Lipschitz continuous;
to streamline our presentation, this will be our standing assumption in what follows.}

Starting with the seminal work of \cite{Ros65}, much of the literature on continuous games and their applications has focused on games that satisfy a condition known as \acdef{DSC}.
In its simplest form, this condition posits that there exist positive constants $\lambda_{\play}>0$ such that
\begin{equation}
\label{eq:DSC}
\tag{DSC}
\sum_{\play\in\players} \lambda_{\play} \braket{\payv_{\play}(\alt \state) - \payv_{\play}(\state)}{\alt \state_{\play} - \state_{\play}}
	< 0
	\quad
	\text{for all $\state,\alt \state\in\feas$, $\state \neq \alt \state$}.
\end{equation}
Owing to the formal similarity between \eqref{eq:DSC} and the various operator monotonicity conditions in optimization \citep[see \eg][]{BC17}, games that satisfy \eqref{eq:DSC} are commonly referred to as (strictly) \emph{monotone}.
As was shown by \citet[Theorem~2]{Ros65}, monotone games admit a unique \acl{NE} $\eq\in\feas$, which, in view of \eqref{eq:DSC} and \eqref{eq:NE}, is also the unique solution of the (weighted) \acl{VI}
\begin{equation}
\label{eq:VI}
\tag{VI}
\sum_{\play\in\players} \lambda_{\play} \braket{\payv_{\play}(\state)}{\state_{\play} - \eq_{\play}}
	< 0
	\quad
	\text{for all $\state\neq\eq$}.
\end{equation}
This property of \aclp{NE} of monotone games will play a crucial role in our analysis and we will use it freely in the rest of our paper.

In terms of applications, monotonicity gives rise to a very rich class of games.
As we show in the paper's supplement, \cref{ex:game-Cournot,ex:game-auction} both satisfy \acl{DSC} (with a nontrivial choice of weights for the latter), as do atomic splittable congestion games in networks with parallel links \citep{ORShi93,SW16}, multi-user covariance matrix optimization problems in \ac{MIMO} systems \citep{MBNS17}, and many other problems where online decision-making is the norm.
Namely, the class of monotone games contains all strictly convex-concave zero-sum games and all games that admit a (strictly) concave \emph{potential}, \ie a function $\pot\from\feas\to\R$ such that
\(
\payv_{\play}(\state)
	= \del_{\play} \pot(\state)
\)
for all $\state\in\feas$, $\play\in\players$.
In view of all this (and unless explicitly stated otherwise),
we will focus throughout on monotone games;
for completeness, we also include in the supplement a straightforward second-order test for monotonicity.

\section{Regularized no-regret learning}
\label{sec:learning}

We now turn to the learning methods that players could employ to increase their individual rewards in an online manner.
Building on \citepos{Zin03} \acl{OGD} policy, the most widely used algorithmic schemes for no-regret learning in the context of online convex optimization invariably revolve around the idea of \emph{regularization}.
To name but the most well-known paradigms, \acf{FTRL} explicitly relies on best-responding to a regularized aggregate of the reward functions revealed up to a given stage, while \acf{OMD} and its variants use a linear surrogate thereof.
All these no-regret policies fall under the general umbrella of ``regularized learning'' and their origins can be traced back to the seminal \acdef{MD} algorithm of \cite{NY83}.%
\footnote{In a utility maximization setting, \acl{MD} should be called mirror \emph{ascent} because players seek to \emph{maximize} their rewards (as opposed to \emph{minimizing} their losses).
Nonetheless, we keep the term ``descent'' throughout because, despite the role reversal, it is the standard name associated with the method.}

The basic idea of \acl{MD} is to generate a new feasible point $\new\state$ by taking a so-called ``mirror step'' from a starting point $\state$ along the direction of an ``approximate gradient'' vector $\dstate$ (which we treat here as an element of the dual space $\dspace\equiv\prod_{\play}\dspace_{\play}$ of $\vecspace\equiv\prod_{\play}\vecspace_{\play}$).%
\footnote{For concreteness (and in a slight abuse of notation), we assume in what follows that $\vecspace$ is equipped with
the product norm $\norm{\state}^{2} = \sum_{\play} \norm{\state_{\play}}^{2}$
and
$\dspace$ with the dual norm $\dnorm{\dstate} = \max\setdef{\braket{\dstate}{\state}}{\norm{\state} \leq 1}$.}
To do so, let $\hreg_{\play}\from\feas_{\play}\to\R$ be a
continuous
and
$\strong_{\play}$-strongly convex \emph{distance-generating} (or \emph{regularizer}) function, \ie
\begin{equation}
\label{eq:strong}
\hreg_{\play}(t\state_{\play} + (1-t) \alt\state_{\play})
	\leq t \hreg_{\play}(\state_{\play})
	+ (1-t) \hreg_{\play}(\alt\state_{\play})
	- \tfrac{1}{2} \strong_{\play} t(1-t) \norm{\alt\state_{\play} -\state_{\play}}^{2},
\end{equation}
for all $\state_{\play},\alt\state_{\play}\in\feas_{\play}$ and all $t\in[0,1]$.
In terms of smoothness (and in a slight abuse of notation) we also assume that the subdifferential of $\hreg_{\play}$ admits a \emph{continuous selection}, \ie a continuous function $\nabla\hreg_{\play}\from\dom\subd\hreg_{\play}\to\dspace_{\play}$ such that $\nabla\hreg_{\play}(\state_{\play})\in\subd\hreg_{\play}(\state_{\play})$ for all $\state_{\play}\in\dom\subd\hreg_{\play}$.%
\footnote{Recall here that the subdifferential of $\hreg_{\play}$ at $\state_{\play}\in\feas_{\play}$ is defined as
\(
\subd\hreg_{\play}(\state_{\play})
	\equiv \setdef
		{\dstate_{\play}\in\dspace_{\play}}
		{\hreg_{\play}(\alt\state_{\play}) \geq \hreg_{\play}(\state_{\play}) + \braket{\dstate_{\play}}{\alt\state_{\play} -\state_{\play}} \;\text{for all}\; \alt\state_{\play}\in\vecspace_{\play}},
\)
with the standard convention that $\hreg_{\play}(\state_{\play}) = +\infty$ if $\state_{\play}\in\vecspace_{\play}\setminus\feas_{\play}$.
By standard results, the domain of subdifferentiability $\subd\hreg_{\play} \equiv \setdef{\state_{\play}\in\feas_{\play}}{\subd\hreg_{\play}\neq\varnothing}$ of $\hreg_{\play}$ satisfies $\intfeas_{\play} \subseteq \dom\subd\hreg_{\play} \subseteq \feas_{\play}$.}
Then, letting $\hreg(\state) = \sum_{\play} \hreg_{\play}(\state_{\play})$ for $\state\in\feas$ (so $\hreg$ is strongly convex with modulus $\strong = \min_{\play} \strong_{\play}$), we get a \emph{pseudo-distance} on $\feas$ via the relation
\begin{equation}
\label{eq:Bregman}
\breg(\base,\state)
	= \hreg(\base) - \hreg(\state) - \braket{\nabla\hreg(\state)}{\base -\state},
\end{equation}
for all $\base\in\feas$, $\state\in\dom\subd\hreg$.

This pseudo-distance is known as the \emph{Bregman divergence} and we have $\breg(\base,\state) \geq 0$ with equality if and only if $\state = \base$;
on the other hand, $\breg$ may fail to be symmetric and/or satisfy the triangle inequality so, in general, it is not a bona fide distance function on $\feas$.
Nevertheless, we also have $\breg(\base,\state) \geq \frac{1}{2} \strong \norm{\state - \base}^{2}$ (see the paper's supplement), so the convergence of a sequence $\act_{\run}$ to $\base$ can be checked by showing that $\breg(\base,\act_{\run}) \to 0$.
For technical reasons, it will be convenient to also assume the converse, \ie that $\breg(\base,\act_{\run}) \to 0$ when $\act_{\run} \to \base$.
This condition is known in the literature as ``Bregman reciprocity'' \citep{CT93}, and it will be our blanket assumption in what follows (note that it is trivially satisfied by \cref{ex:Eucl,ex:entropic} below).

Now, as with true Euclidean distances, $\breg(\base,\state)$ induces a \emph{prox-mapping} given by
\begin{equation}
\label{eq:prox}
\prox_{\state}(\dstate)
	= \argmin_{\alt\state\in\feas} \{\braket{\dstate}{\state - \alt\state} + \breg(\alt\state,\state)\}
\end{equation}
for all $\state\in\dom\subd\hreg$ and all $\dstate\in\dspace$.
Just like its Euclidean counterpart below, the prox-mapping \eqref{eq:prox} starts with a point $\state\in\dom\subd\hreg$ and steps along the dual (gradient-like) vector $\dstate\in\dspace$ to produce a new feasible point $\new\state = \prox_{\state}(\dstate)$.
Standard examples of this process are:
\smallskip

\begin{example}[Euclidean projections]
\label{ex:Eucl}
Let $\hreg(\state) = \frac{1}{2} \norm{\state}_{2}^{2}$ denote the Euclidean squared norm.
Then,
the induced prox-mapping is
\begin{equation}
\label{eq:prox-Eucl}
\prox_{\state}(\dstate)
	= \Eucl(\state + \dstate),
\end{equation}
with $\Eucl(\state) = \argmin_{\alt\state\in\feas} \norm{\alt\state -\state}^{2}$ denoting the standard Euclidean projection onto $\feas$.
Hence, the update rule $\new\state = \prox_{\state}(\dstate)$ boils down to a ``vanilla'', Euclidean projection step along $\dstate$.
\end{example}
\smallskip

\begin{example}[Entropic regularization and multiplicative weights]
\label{ex:entropic}
Suppressing the player index for simplicity, let $\feas$ be a $\vdim$-dimensional simplex and consider the entropic regularizer $\hreg(\state) = \sum_{j=1}^{\vdim} \state_{j} \log \state_{j}$.
The induced pseudo-distance is the so-called \acdef{KL} divergence
\(
\dkl(\base,\state)
	= \sum_{j=1}^{\vdim} \base_{j} \log \parens{\base_{j} / \state_{j}},
\)
which gives rise to the prox-mapping
\begin{equation}
\label{eq:prox-entropic}
\prox_{\state}(\dstate)
	= \frac{(\state_{j} \exp(\dstate_{j}))_{j=1}^{\vdim}}{\sum_{j=1}^{\vdim} \state_{j} \exp(\dstate_{j})}
\end{equation}
for all $\state\in\intfeas$, $\dstate\in\dspace$.
The update rule $\new\state = \prox_{\state}(\dstate)$ is widely known as the \acdef{MW} algorithm and plays a central role for learning in multi-armed bandit problems and finite games \citep{ACBFS95,FS99,AHK12}.
\end{example}
\smallskip

With all this in hand, the multi-agent \acdef{MD} algorithm is given by the recursion
\begin{equation}
\label{eq:MD}
\tag{MD}
\act_{\run+1}
	= \prox_{\act_{\run}}(\step_{\run} \est\payv_{\run}),
\end{equation}
where $\step_{\run}$ is a variable step-size sequence
and
$\est\payv_{\run} = (\est\payv_{\play,\run})_{\play\in\players}$ is a generic feedback sequence of estimated gradients.
In the next section, we detail how this sequence is generated with first- or zeroth-order (bandit) feedback.

%
%

\section{First-order vs. bandit feedback}
\label{sec:feedback}

\subsection{First-order feedback\afterhead}

A common assumption in the literature is that players are able to obtain gradient information by querying a \emph{first-order oracle} \citep{Nes04}.
\ie a ``black-box'' feedback mechanism that outputs an estimate $\est\payv_{\play}$ of the individual payoff gradient $\payv_{\play}(\state)$ of the $\play$-th player at the current action profile $\state = (\state_{\play};\state_{-\play})\in\feas$.
This estimate could be either \emph{perfect}, giving $\est\payv_{\play} = \payv_{\play}(\state)$ for all $\play\in\players$, or \emph{imperfect}, returning noisy information of the form
\(
\est\payv_{\play}
	= \payv_{\play}(\state) + \noise_{\play}
\)
where $\noise_{\play}$ denotes the oracle's error (random, systematic, or otherwise).

Having access to a perfect oracle is usually a tall order,
either because payoff gradients are difficult to compute directly (especially without global knowledge),
because they involve an expectation over a possibly unknown probability law,
or for any other number of reasons.
It is therefore more common to assume that each player has access to a \emph{stochastic oracle} which, when called against a sequence of actions $\act_{\run}\in\feas$, produces a sequence of gradient estimates $\est\payv_{\run} = (\payv_{\play,\run})_{\play\in\players}$ that satisfies the following statistical assumptions:
\begin{equation}
\label{eq:oracle}
\begin{aligned}
&a)\;
	\textit{Unbiasedness:}
	&
	&\exof{\est\payv_{\run} \given \filter_{\run}}
		= \payv(\act_{\run}).
	\\
&b)\;
	\textit{Finite mean square:}
	&
	&\exof{\dnorm{\est\payv_{\run}}^{2} \given \filter_{\run}}
		\leq \vbound^{2}
		\;\;
		\text{for some finite $\vbound\geq0$}.
		\hspace{6em}
\end{aligned}
\end{equation}
In terms of measurability, the expectation in \eqref{eq:oracle} is conditioned on the history $\filter_{\run}$ of
$\act_{\run}$ up to stage $\run$;
in particular, since $\est\payv_{\run}$ is generated randomly from $\act_{\run}$, it is not $\filter_{\run}$-measurable (and hence not adapted).
To make this more transparent, we will write
\(
\est\payv_{\run}
	= \payv(\act_{\run}) + \noise_{\run+1}
\)
where $\noise_{\run}$ is
an adapted \acl{MDS} with $\exof{\dnorm{\noise_{\run+1}}^{2} \given \filter_{\run}} \leq \noisevar$ for some finite $\noisedev \geq 0$.
\smallskip


\subsection{Bandit feedback\afterhead}

Now, if players don't have access to a first-order oracle \textendash\ the so-called \emph{bandit} or \emph{payoff-based} framework \textendash\ they will need to derive an individual gradient estimate from the only information at their disposal:
the actual payoffs they receive at each stage.
When a function can be queried at multiple points (as few as two in practice), there are efficient ways to estimate its gradient via directional sampling techniques as in \cite{ADX10}.
In a game-theoretic setting however, multiple-point estimation techniques do not apply because,
in general, a player's payoff function depends on the actions of \emph{all} players.
Thus, when a player attempts to get a second query of their payoff function, this function may have already changed due to the query of another player \textendash\ \ie instead of sampling $\pay_{\play}(\cdot;\state_{-\play})$, the $\play$-th player would be sampling $\pay_{\play}(\cdot;\alt \state_{-\play})$ for some $\alt\state_{-\play} \neq \state_{-\play}$.

Following \cite{Spa97} and \cite{FKM05}, we posit instead that players rely on a \acf{SPSA} approach that allows them to estimate their individual payoff gradients $\payv_{\play}$ based off a \emph{single} function evaluation.
In detail, the key steps of this one-shot estimation process for each player $\play\in\players$ are:
\begin{enumerate}
\setcounter{enumi}{-1}
\item
Fix a \emph{query radius} $\mix>0$.%
\footnote{For simplicity, we take $\mix$ equal for all players;
the extension to player-specific $\mix$ is straightforward, so we omit it.}
\item
Pick a \emph{pivot point} $\state_{\play}\in\feas_{\play}$ where player $\play$ seeks to estimate their payoff gradient.
\item
Draw a vector $\unitvec_{\play}$ from the unit sphere $\sphere_{\play}\equiv\sphere^{\vdim_{\play}}$ of $\vecspace_{\play}\equiv\R^{\vdim_{\play}}$ and play $\est \state_{\play} = \state_{\play} + \mix\unitvec_{\play}$.%
\footnote{We tacitly assume here that the query directions $\unitvec_{\play}\in\sphere^{\vdim_{\play}}$ are drawn independently across players.}
\item
Receive $\est\pay_{\play} = \pay_{\play}(\est \state_{\play};\est \state_{-\play})$ and set
\begin{equation}
\label{eq:SPSA}
\est\payv_{\play}
	= \frac{\vdim_{\play}}{\mix} \est\pay_{\play} \, \unitvec_{\play}.
\end{equation}
\end{enumerate}

By adapting a standard argument based on Stokes' theorem (detailed in the supplement),
it can be shown that $\est\payv_{\play}$ is an unbiased estimator of the individual gradient of the $\mix$-smoothed payoff function
\begin{equation}
\label{eq:smoothed}
\pay_{\play}^{\mix}(\state)
	= \frac{1}{\vol(\mix\ball_{\play}) \prod_{\playalt\neq\play} \vol(\mix\sphere_{\playalt})}
	\int_{\mix\ball_{\play}} \int_{\prod_{\playalt\neq\play} \mix\sphere_{\playalt}}
		\pay_{\play}(\state_{\play} + \pert_{\play};\state_{-\play} + \unitvec_{-\play})\,
		\dd\unitvec_{1} \dotsm d\pert_{\play} \dotsm d\unitvec_{\nPlayers}
\end{equation}
with $\ball_{\play}\equiv\ball^{\vdim_{\play}}$ denoting the unit ball of $\vecspace_{\play}$.
The Lipschitz continuity of $\payv_{\play}$ guarantees that $\norm{\del_{\play}\pay_{\play} - \del_{\play}\pay_{\play}^{\mix}}_{\infty} = \bigoh(\mix)$,
so this estimate becomes more and more accurate as $\mix\to0^{+}$.
On the other hand, the second moment of $\est\payv_{\play}$ grows as
\(
	\bigoh\parens{1/\mix^{2}},
\)
implying in turn that the variability of $\est\payv_{\play}$ grows unbounded as $\mix\to0^{+}$.
This manifestation of the bias-variance dilemma plays a crucial role in designing no-regret policies with bandit feedback \citep{FKM05,Kle04}, so $\mix$ must be chosen with care.

Before dealing with this choice though, it is important to highlight two feasibility issues that arise with the single-shot \ac{SPSA} estimate \eqref{eq:SPSA}.
The first has to do with the fact that the perturbation direction $\unitvec_{\play}$ is chosen from the unit sphere $\sphere_{\play}$ so it may fail to be tangent to $\feas_{\play}$, even when $\state_{\play}$ is interior.
To iron out this wrinkle, it suffices to sample $\unitvec_{\play}$ from the intersection of $\sphere_{\play}$ with the affine hull of $\feas_{\play}$ in $\vecspace_{\play}$;
on that account (and without loss of generality), we will simply assume in what follows that each $\feas_{\play}$ is a \emph{convex body} of $\vecspace_{\play}$, \ie it has nonempty topological interior.

The second feasibility issue concerns the size of the perturbation step:
even if $\unitvec_{\play}$ is a feasible direction of motion, the query point $\est \state_{\play} = \state_{\play} + \mix\unitvec_{\play}$ may be unfeasible if $\state_{\play}$ is too close to the boundary of $\feas_{\play}$.
For this reason, we will introduce a ``safety net'' in the spirit of \cite{ADX10}, and we will constrain the set of possible pivot points $\state_{\play}$ to lie within a suitably shrunk zone of $\feas$.

In detail, let $\ball_{\radius_{\play}}(\base_{\play})$ be an $\radius_{\play}$-ball centered at $\base_{\play}\in\feas_{\play}$ so that $\ball_{\radius_{\play}}(\base_{\play})\subseteq\feas_{\play}$.
Then, instead of perturbing $\state_{\play}$ by $\unitvec_{\play}$, we consider the \emph{feasibility adjustment}
\begin{equation}
\label{eq:pert}
\pert_{\play}
	= \unitvec_{\play} - \radius_{\play}^{-1}(\state_{\play} - \base_{\play}),
\end{equation}
and each player plays $\est \state_{\play} = \state_{\play} + \mix\pert_{\play}$ instead of $\state_{\play} + \mix\unitvec_{\play}$.
In other words, this adjustment moves each pivot to
$\state_{\play}^{\mix} = \state_{\play} - \radius_{\play}^{-1}\mix (\state_{\play} - \base_{\play})$,
\ie $\bigoh(\mix)$-closer to the interior base point $\base_{\play}$, and then perturbs $\state_{\play}^{\mix}$ by $\mix\unitvec_{\play}$.
Feasibility of the query point is then ensured by noting that
\begin{equation}
\label{eq:query}
\est \state_{\play}
	= \state_{\play}^{\mix} + \mix\unitvec_{\play}
	= (1 - \radius_{\play}^{-1}\mix) \state_{\play}
	+ \radius_{\play}^{-1}\mix (\base_{\play} + \radius_{\play} \unitvec_{\play}),
\end{equation}
so $\est \state_{\play}\in\feas_{\play}$ if $\mix/\radius_{\play}<1$ (since $\base_{\play} + \radius_{\play}\unitvec_{\play} \in \ball_{\radius_{\play}}(\base_{\play}) \subseteq \feas_{\play}$).

The difference between this estimator and the oracle framework we discussed above is twofold.
First, each player's \emph{realized} action is $\est \state_{\play} = \state_{\play} + \mix\pert_{\play}$, not $\state_{\play}$, so there is a disparity between the point at which payoffs are queried and the action profile where the oracle is called.
Second, the resulting estimator $\hat\payv$ is not unbiased, so the statistical assumptions \eqref{eq:oracle} for a stochastic oracle do not hold.
In particular, given the feasibility adjustment \eqref{eq:pert}, the estimate \eqref{eq:SPSA} with $\est \state$ given by \eqref{eq:query} satisfies
\begin{equation}
\label{eq:bias}
\exof{\est \payv_{\play}}
	= \del_{\play} \pay_{\play}^{\mix}(\state_{\play}^{\mix};\state_{-\play}^{\mix}),
\end{equation}
so there are \emph{two} sources of systematic error:
an $\bigoh(\mix)$ perturbation in the function, and an $\bigoh(\mix)$ perturbation of each player's pivot point from $\state_{\play}$ to $\state_{\play}^{\mix}$.
Hence, to capture both sources of bias and separate them from the random noise, we will write
\begin{equation}
\label{eq:grad-est}
\est\payv_{\play}
	= \payv_{\play}(\state)
	+ \noise_{\play}
	+ \bias_{\play}
\end{equation}
where $\noise_{\play} = \est\payv_{\play} - \exof{\est\payv_{\play}}$ and $\bias_{\play} = \del_{\play} \pay_{\play}^{\mix}(\state^{\mix}) - \del_{\play}\pay_{\play}(\state)$.
We are thus led to the following manifestation of the bias-variance dilemma:
the bias term $\bias$ in \eqref{eq:grad-est} is $\bigoh(\mix)$, but the second moment of the noise term $\noise$ is $\bigoh(1/\mix^{2})$;
as such, an increase in accuracy (small bias) would result in a commensurate loss of precision (large noise variance).
Balancing these two factors will be a key component of our analysis.

\section{Convergence analysis and results}
\label{sec:results}


\begin{algorithm}[tbp]
\caption{Multi-agent \acl{MD} with bandit feedback\; (player indices suppressed)}
\label{alg:MD-0}

\small
\tt
\begin{algorithmic}[1]
\Require
	step-size $\step_{\run}>0$,
	query radius $\mix_{\run}>0$,
	safety ball $\ball_{\radius}(\base)\subseteq\feas$
\State
	choose $\act\in\dom\subd\hreg$
	\Comment{initialization}%
\Repeat\;
	at each stage $\run=\running$
	\State	
		draw $\unitvar$ uniformly from $\sphere^{\vdim}$
		\Comment{perturbation direction}%
	\State
		set $\pertvar\leftarrow\unitvar - \radius^{-1}(\act - \base)$
		\Comment{query direction}%
	\State
		play $\est\act \leftarrow \act + \mix_{\run}\pertvar$
		\Comment{choose action}%
	\State
		receive $\est\pay \leftarrow \pay(\est\act)$
		\Comment{get payoff}%
	\State
		set $\est\payv\leftarrow (\vdim/\mix_{\run}) \est\pay \cdot \unitvar$
		\Comment{estimate gradient}%
	\State
		update $\act \leftarrow \prox_{\act}(\step_{\run}\est\payv)$
		\Comment{update pivot}%
\Until{end}
\end{algorithmic}
\end{algorithm}


Combining the learning framework of \cref{sec:learning} with the single-shot gradient estimation machinery of \cref{sec:feedback}, we obtain the following variant of \eqref{eq:MD} with payoff-based, \emph{bandit feedback:}
\begin{equation}
\label{eq:MD-0}
\tag{\ref*{eq:MD}-b}
\begin{aligned}
\est\act_{\run}
	&= \act_{\run} + \mix_{\run} \pertvar_{\run},
	\\
\act_{\run+1}
	&= \prox_{\act_{\run}}(\step_{\run}\est\payv_{\run}).
\end{aligned}
\end{equation}
In the above, the perturbations $\pertvar_{\run}$ and the estimates $\est\payv_{\run}$ are given respectively by \eqref{eq:pert} and \eqref{eq:SPSA}, \ie
\begin{equation}
\pertvar_{\play,\run}
	= \unitvar_{\play,\run} - \radius_{\play}^{-1} \parens{\act_{\play,\run} - \base_{\play}}
	\qquad
\est\payv_{\play,\run}
	= (\vdim_{\play}/\mix_{\run}) \pay_{\play}(\est\act_{\run}) \, \unitvar_{\play,\run}
\end{equation}
and $\unitvar_{\play,\run}$ is drawn independently and uniformly across players at each stage $\run$
(see also \cref{alg:MD-0} for a pseudocode implementation and \cref{fig:MD-0} for a schematic representation).

In the rest of this paper, our goal will be to determine the equilibrium convergence properties of this scheme in concave $\nPlayers$-person games.
Our first asymptotic result below shows that, under \eqref{eq:MD-0}, the players' learning process converges to \acl{NE} in monotone games:


\begin{figure}[tbp]
\centering

\begin{tikzpicture}
[scale=1,
nodestyle/.style = {circle,fill=black,inner sep=0, minimum size=2.5}]
\footnotesize
\def\angleOne{5}
\def\angleTwo{-50}

\draw [thick, color = MyBlue, pattern = north east lines, pattern color = MyBlue!10] (105:3) arc (105:-75:3);
\draw node at (.5,1) {\normalsize$\feas$};

\node [nodestyle, label = above:{$\act_{1}$}] (pivot1) at (60:2.25) {};
\node (score1) at (\angleOne:4.5) {};
\node [nodestyle, label = above:$\est\act_{1}$] (act1) at ($.75*(pivot1) + 0.25*(score1)$) {.};
\node [nodestyle] (score1) at (score1) {.};
\draw[-stealth,BrickRed] (pivot1.center) -- (score1.center) node [midway, right]{$\;\;\;\step(\vdim/\mix)\, \est\pay_{1} \unitvec_{1}$};
\draw[-stealth] (pivot1.center) -- (act1.center) node [midway, below]{$\mix \unitvec_{1}$};

\node [nodestyle, label = above left:$\act_{2}$] (pivot2) at (\angleOne:3) {.};
\draw[-stealth, densely dashed] (score1.center) -- (pivot2.center) node [midway, below]{$\Eucl$};

\node (score2) at (\angleTwo:2) {};
\node [nodestyle, label = below:$\quad\est\act_{2}$] (act2) at ($(pivot2) + 0.35*(score2) - 0.35*(pivot2)$) {.};
\node [nodestyle, label = below:$\act_{3}$] (score2) at (score2) {.};
\draw[-stealth,BrickRed] (pivot2.center) -- (score2.center) node [midway, left]{$\step(\vdim/\mix)\, \est\pay_{2} \unitvec_{2}$};
\draw[-stealth] (pivot2.center) -- (act2.center) node [midway, left]{$\mix \unitvec_{2}\;$};



\end{tikzpicture}
\caption{Schematic representation of \cref{alg:MD-0} with ordinary, Euclidean projections.
To reduce visual clutter, we did not include the feasibility adjustment $\radius^{-1}(\state - \base)$ in the action selection step $\act_{\run} \mapsto \est\act_{\run}$.}
\label{fig:MD-0}
\end{figure}


\begin{theorem}
\label{thm:convergence}
Suppose that the players of a monotone game $\game\equiv\game(\players,\feas,\pay)$ follow \eqref{eq:MD-0} with step-size $\step_{\run}$ and query radius $\mix_{\run}$ such that
\begin{equation}
\label{eq:params}
\lim_{\run\to\infty} \step_{\run}
	= \lim_{\run\to\infty} \mix_{\run}
	= 0,
	\quad
\sum_{\run=\start}^{\infty} \step_{\run}
	= \infty,
	\quad
\sum_{\run=\start}^{\infty} \step_{\run} \mix_{\run}
	< \infty,
	\quad
\text{and}
	\quad
\sum_{\run=\start}^{\infty} \frac{\step_{\run}^{2}}{\mix_{\run}^{2}}
	< \infty.
\end{equation}
Then, the sequence of realized actions $\est\act_{\run}$ converges to \acl{NE} with probability $1$.
\end{theorem}

Even though the setting is different, the conditions \eqref{eq:params} for the tuning of the algorithm's parameters are akin to those encountered in \acl{KW} stochastic approximation schemes and serve a similar purpose.
First, the conditions $\lim_{\run\to\infty} \step_{\run} = 0$ and $\sum_{\run=\start}^{\infty} \step_{\run} = \infty$ respectively mitigate the method's inherent randomness and ensure a horizon of sufficient length.
The requirement $\lim_{\run\to\infty} \mix_{\run} = 0$ is also straightforward to explain:
as players accrue more information, they need to decrease the sampling bias in order to have any hope of converging.
However, as we discussed in \cref{sec:feedback}, decreasing $\mix$ also increases the variance of the players' gradient estimates, which might grow to infinity as $\mix\to0$.
The crucial observation here is that new gradients enter the algorithm with a weight of $\step_{\run}$ so the aggregate bias after $\run$ stages is of the order of $\bigoh(\sum_{\runalt=\start}^{\run} \step_{\runalt}\mix_{\runalt})$ and its variance is $\bigoh(\sum_{\runalt=\start}^{\run} \step_{\runalt}^{2}/\mix_{\runalt}^{2})$.
If these error terms can be controlled, there is an underlying drift that emerges over time and which steers the process to equilibrium.
We make this precise in the supplement by using a suitably adjusted variant of the Bregman divergence as a quasi-Féjér energy function for \eqref{eq:MD-0} and relying on a series of (sub)martingale convergence arguments to establish the convergence of $\est\act_{\run}$ (first as a subsequence, then with probability $1$).


Of course, since \cref{thm:convergence} is asymptotic in nature, it is not clear how to choose $\step_{\run}$ and $\mix_{\run}$ so as to optimize the method's convergence rate.
Heuristically, if we take schedules of the form $\step_{\run} = \step/\run^{\pexp}$ and $\mix_{\run} = \mix/\run^{\qexp}$ with $\step,\mix>0$ and $0 < \pexp,\qexp \leq 1$, the only conditions imposed by \eqref{eq:params} are $\pexp+\qexp>1$ and $\pexp-\qexp>1/2$.
However, as we discussed above, the aggregate bias in the algorithm after $\run$ stages is $\bigoh(\sum_{\runalt=\start}^{\run} \step_{\run}\mix_{\run}) = \bigoh(1/\run^{\pexp+\qexp-1})$ and its variance is $\bigoh(\sum_{\runalt=\start}^{\run} \step_{\runalt}^{2}/\mix_{\runalt}^{2}) = \bigoh(1/\run^{2\pexp-2\qexp-1})$:
if the conditions \eqref{eq:params} are satisfied, both error terms vanish, but they might do so at very different rates.
By equating these exponents in order to bridge this gap, we obtain $\qexp = \pexp/3$;
moreover, since the single-shot \ac{SPSA} estimator \eqref{eq:SPSA} introduces a $\Theta(\mix_{\run})$ random perturbation, $\qexp$ should be taken as large as possible to ensure that this perturbation vanishes at the fastest possible rate.
As a result, the most suitable choice for $\pexp$ and $\qexp$ seems to be $\pexp=1$, $\qexp=1/3$, leading to an error bound of $\bigoh(1/\run^{1/3})$.

We show below that this bound is indeed attainable for games
that are \emph{strongly monotone}, \ie they
satisfy the following stronger variant of \acl{DSC}:
\begin{equation}
\label{eq:DSC-str}
\tag{$\str$-DSC}
\sum_{\play\in\players} \lambda_{\play} \braket{\payv_{\play}(\alt \state) - \payv_{\play}(\state)}{\alt \state_{\play} - \state_{\play}}
	\leq -\frac{\str}{2} \norm{\state - \alt\state}^{2}
\end{equation}
for some $\lambda_{\play},\str>0$ and for all $\state,\alt \state\in\feas$.
Focusing for expository reasons on the most widely used, Euclidean incarnation of the method (\cref{ex:Eucl}),
we have:

\begin{theorem}
\label{thm:rate}
Let $\eq$ be the \textpar{necessarily unique} \acl{NE} of a $\str$-strongly monotone game.
If the players follow \eqref{eq:MD-0} with Euclidean projections and parameters $\step_{\run} = \step/\run$ and $\mix_{\run} = \mix/\run^{1/3}$ with $\step > 1/(3\str)$ and $\mix>0$, we have
\begin{equation}
\label{eq:rate}
\exof{\norm{\est\act_{\run} - \eq}^{2}}
	= \bigoh(\run^{-1/3}).
\end{equation}
\end{theorem}

\cref{thm:rate} is our main finite-time analysis result, so some remarks are in order.
First, the step-size schedule $\step_{\run} \propto 1/\run$ is not required to obtain an $\bigoh(\run^{-1/3})$ convergence rate:
as we show in the paper's supplement, more general schedules of the form $\step_{\run} \propto 1/\run^{\pexp}$ and $\mix_{\run} \propto 1/\run^{\qexp}$ with $\pexp>3/4$ and $\qexp=\pexp/3 > 1/4$, still guarantee an $\bigoh(\run^{-1/3})$ rate of convergence for \eqref{eq:MD-0}.
To put things in perspective, we also show in the supplement that if \eqref{eq:MD} is run with first-order oracle feedback satisfying the statistical assumptions \eqref{eq:oracle}, the rate of convergence becomes $\bigoh(1/\run)$.
Viewed in this light, the price for not having access to gradient information is no higher than $\bigoh(\run^{-2/3})$ in terms of the players' equilibration rate.


Finally, it is also worth comparing the bound \eqref{eq:rate} to the attainable rates for stochastic convex optimization (the single-player case).
For problems with objectives that are both strongly convex and smooth,
\cite{ADX10} attained an $\bigoh(\run^{-1/2})$ convergence rate with bandit feedback,
which \cite{Sha13} showed is unimprovable.
Thus, in the single-player case, the bound \eqref{eq:rate} is off by $\run^{1/6}$ and coincides with the bound of \cite{ADX10} for strongly convex functions that are not necessarily smooth.
One reason for this gap is that the $\Theta(\run^{-1/2})$ bound of \cite{Sha13} concerns the smoothed-out time average $\bar\act_{\run} = \run^{-1} \sum_{\runalt=1}^{\run} \act_{\runalt}$, while our analysis concerns the sequence of \emph{realized actions} $\est\act_{\run}$.
This difference is semantically significant:
In optimization, the query sequence is just a means to an end, and only the algorithm's output matters
(\ie $\bar\act_{\run}$).
In a game-theoretic setting however, it is the players' \emph{realized} actions that determine their rewards at each stage, so the figure of merit is the actual sequence of play $\est\act_{\run}$.
This sequence
is more difficult to control, so this disparity is, perhaps, not too surprising;
nevertheless, we believe that this gap can be closed by using a more sophisticated single-shot estimate, \eg as in \cite{GL13}.
We defer this analysis to the future.

\section{Concluding remarks}
\label{sec:conclusions}

The most sensible choice for agents who are oblivious to the presence of each other (or who are simply conservative), is to deploy a no-regret learning algorithm.
With this in mind, we studied the long-run behavior of individual regularized no-regret learning policies and we showed that, in monotone games,
play converges to equilibrium with probability $1$, and the rate of convergence almost matches the optimal rates of \emph{single-agent}, stochastic convex optimization.
Nevertheless, several questions remain open: 
whether there is an intrinsic information-theoretic obstacle to bridging this gap;
whether our convergence rate estimates hold with high probability (and not just in expectation);
and
whether our analysis 
extends to a fully decentralized setting where the players' updates need not be synchronous.
We intend to address these questions in future work.

\appendix

\section{Monotone games}
\label{app:monotone}

Our aim in this appendix is to show that the game-theoretic examples of \cref{sec:prelims} are both monotone.
Before studying them in detail, it will be convenient to introduce a straightforward second-order test for monotonicity based on the game's Hessian matrix.

Specifically, extending the notion of the Hessian of an ordinary (scalar) function, the \emph{\textpar{$\lambda$-weighted} Hessian} of a game $\game$ is defined as the block matrix $H_{\game}(\state;\lambda) = (H_{\play\playalt}(\state;\lambda))_{\play,\playalt\in\players}$ with blocks
\begin{equation}
\label{eq:Hessian}
H_{\play\playalt}(\state;\lambda)
	= \frac{\lambda_{\play}}{2} \del_{\playalt} \del_{\play} \pay_{\play}(\state)
	+ \frac{\lambda_{\playalt}}{2} \parens*{\del_{\play} \del_{\playalt} \pay_{\playalt}(\state)}^{\top}.
\end{equation}
As was shown by \citet[Theorem~6]{Ros65}, $\game$ satisifes \eqref{eq:DSC} with weight vector $\lambda$ whenever $z^{\top} H_{\game}(\state;\lambda) z < 0$ for all $\state\in\feas$ and all nonzero $z\in\vecspace \equiv \prod_{\play} \vecspace_{\play}$ that are tangent to $\feas$ at $\state$.%
\footnote{By ``tangent'' we mean here that $z$ belongs to the tangent cone $\tcone(\state)$ to $\feas$ at $\state$, \ie the intersection of all supporting (closed) half-spaces of $\feas$ at $\state$.}
It is thus common to check for monotonicity by taking $\lambda_{\play}=1$ for all $\play\in\players$ and verifying whether the unweighted Hessian of $\game$ is negative-definite on the affine hull of $\feas$.

\subsection{Cournot competition (Example \ref{ex:game-Cournot})\afterhead}
In the standard Cournot oligopoly model described in the main body of the paper, the players' payoff functions are given by
\begin{equation}
\txs
\pay_{\play}(\state)
	= \state_{\play} \parens[\big]{a - b\sum_{\playalt} \state_{\playalt}}
	- c_{\play} \state_{\play}.
\end{equation}
Consequently, a simple differentiation yields
\begin{equation}
H_{\play\playalt}(\state)
	= \frac{1}{2} \frac{\pd^{2} \pay_{\play}}{\pd\state_{\play}\pd\state_{\playalt}}
	+ \frac{1}{2} \frac{\pd^{2} \pay_{\playalt}}{\pd\state_{\playalt}\pd\state_{\play}}
	= -b(1 + \delta_{\play\playalt}),
\end{equation}
where $\delta_{\play\playalt} = \one\{\play=\playalt\}$ is the Kronecker delta.
This matrix is clearly negative-definite, so the game is monotone.

\subsection{Resource allocation auctions (Example \ref{ex:game-auction})\afterhead}
In our auction-theoretic example, the players' payoff functions are given by
\begin{equation}
\label{eq:pay-auction}
\pay_{\play}(\state_{\play};\state_{-\play})
	= \sum_{\resource\in\resources} \bracks*{
	\frac{g_{\play} q_{\resource} \state_{\play\resource}}{c_{\resource} + \sum_{\playalt\in\players} \state_{\playalt\resource}}
	- \state_{\play\resource}}
\end{equation}

To prove monotonicity in this example, we will consider the following criterion due to \cite{Goo80}:
a game $\game$ satisfies \eqref{eq:DSC} with weights $\lambda_{\play}$, $\play\in\players$, if:
\begin{enumerate}
[\hspace{2em}\itshape a\upshape)]
\item
Each payoff function $\pay_{\play}$ is strictly concave in $\state_{\play}$ and convex in $\state_{-\play}$.
\item
The function $\sum_{\play\in\players} \lambda_{\play} \pay_{\play}(\state)$ is concave in $\state$.
\end{enumerate}

Since the function $\phi(x) = x/(c+x)$ is strictly concave in $x$ for all $c>0$, the first condition above is trivial to verify.
For the second, letting $\lambda_{\play} = 1/g_{\play}$ gives
\begin{flalign}
\sum_{\play\in\players} \lambda_{\play} \pay_{\play}(\state)
	&= \sum_{\play\in\players} \sum_{\resource\in\resources}
	\frac{q_{\resource} \state_{\play\resource}}{c_{\resource} + \sum_{\playalt\in\players} \state_{\playalt\resource}}
	- \sum_{\play\in\players} \sum_{\resource\in\resources} \state_{\play\resource}
	\notag\\
	&= \sum_{\resource\in\resources} q_{\resource} \frac{\sum_{\play\in\players} \state_{\play\resource}}{c_{\resource} + \sum_{\play\in\players} \state_{\play\resource}}
	- \sum_{\play\in\players} \sum_{\resource\in\resources} \state_{\play\resource}.
\end{flalign}
Since the summands above are all concave in their respective arguments, our claim follows.

\section{Properties of Bregman proximal mappings}
\label{app:Bregman}

In this appendix, we provide some auxiliary results and estimates that are used throughout the convergence analysis of \cref{app:convergence}.
Some of the results we present here are not new \cite[see \eg][]{NJLS09};
however, the set of hypotheses used to obtain them varies widely in the literature, so we provide all proofs for completeness.

In what follows, we will make frequent use of the convex conjugate $\hreg^{\ast}\from\dspace\to\R$ of $\hreg$, defined here as
\begin{equation}
\label{eq:conj}
\hreg^{\ast}(\dstate)
	= \max_{\state\in\feas} \{ \braket{\dstate}{\state} - \hreg(\state) \}.
\end{equation}
By standard results in convex analysis \cite[Chap.~26]{Roc70}, $\hreg^{\ast}$ is differentiable on $\dspace$ and its gradient satisfies the identity
\begin{equation}
\label{eq:dconj}
\nabla\hreg^{\ast}(\dstate)
	= \argmax_{\state\in\feas} \{ \braket{\dstate}{\state} - \hreg(\state) \}.
\end{equation}
For notational convenience, we will also write
\begin{equation}
\label{eq:mirror}
\mirror(\dstate)
	= \nabla\hreg^{\ast}(\dstate)
\end{equation}
and we will refer to $\mirror\from\dspace\to\feas$ as the \emph{mirror map} generated by $\hreg$.

Together with the prox-mapping induced by $\hreg$, all these notions are related as follows:
\begin{lemma}
\label{lem:mirror}
Let $\hreg$ be a regularizer on $\feas$.
Then, for all $\state\in\dom\subd\hreg$, $\dstate\in\dspace$, we have:
\begin{subequations}
\label{eq:links}
\begin{alignat}{4}
\label{eq:links-mirror}
&a)&
	&\;\;
	\state = \mirror(\dstate)
	&\;\iff\;
	&\dstate \in \subd\hreg(\state).
	\\
\label{eq:links-prox}
&b)&
	&\;\;
	\new\state = \prox_{\state}(\dstate)
	&\;\iff\;
	&\nabla\hreg(\state) + \dstate \in \subd\hreg(\new\state)
	&\;\iff\;
	&\new\state = \mirror(\nabla\hreg(\state) + \dstate).
	\hspace{4em}
\end{alignat}
\end{subequations}
Finally, if $\state = \mirror(\dstate)$ and $\base\in\feas$, we have
\begin{equation}
\label{eq:selection}
\braket{\nabla\hreg(\state)}{\state - \base}
	\leq \braket{\dstate}{\state - \base}.
\end{equation}
\end{lemma}

\begin{remark*}
Note that \eqref{eq:links-prox} directly implies that $\subd\hreg(\new\state) \neq \varnothing$, \ie $\new\state \in \dom\subd\hreg$.
An immediate consequence of this is that the update rule $\state \leftarrow \prox_{\state}(\dstate)$ is \emph{well-posed},
\ie it can be iterated in perpetuity.
\end{remark*}

\begin{proof}[Proof of \cref{lem:mirror}]
To prove \eqref{eq:links-mirror}, note that $\state$ solves \eqref{eq:dconj} if and only if $\dstate - \subd\hreg(\state) \ni 0$, \ie if and only if $\dstate\in\subd\hreg(\state)$.
Similarly, for \eqref{eq:links-prox}, comparing \eqref{eq:prox} and \eqref{eq:conj}, we see that $\new\state$ solves \eqref{eq:prox} if and only if $\nabla\hreg(\state) + \dstate \in \subd\hreg(\new\state)$, \ie if and only if $\new\state = \mirror(\nabla\hreg(\state) + \dstate)$.

For the inequality \eqref{eq:selection}, it suffices to show it holds for interior $\base\in\intfeas$ (by continuity).
To do so, let
\begin{equation}
\phi(t)
	= \hreg(\state + t(\base-\state))
	- \bracks{\hreg(\state) +  \braket{\dstate}{\state + t(\base-\state)}}.
\end{equation}
Since $\hreg$ is strongly convex and $\dstate\in\subd\hreg(\state)$ by \eqref{eq:links-mirror}, it follows that $\phi(t)\geq0$ with equality if and only if $t=0$.
Moreover, note that $\psi(t) = \braket{\nabla\hreg(\state + t(\base-\state)) - \dstate}{\base - \state}$ is a continuous selection of subgradients of $\phi$.
Given that $\phi$ and $\psi$ are both continuous on $[0,1]$, it follows that $\phi$ is continuously differentiable and $\phi' = \psi$ on $[0,1]$.
Thus, with $\phi$ convex and $\phi(t) \geq 0 = \phi(0)$ for all $t\in[0,1]$, we conclude that $\phi'(0) = \braket{\nabla\hreg(\state) - \dstate}{\base - \state} \geq 0$, from which our claim follows.
\end{proof}

We continue with some basic relations connecting the Bregman divergence relative to a target point before and after a prox step.
The basic ingredient for this is a generalization of the law of cosines which is known in the literature as the ``three-point identity'' \citep{CT93}:

\begin{lemma}
\label{lem:3points}
Let $\hreg$ be a regularizer on $\feas$.
Then, for all $\base\in\feas$ and all $\state,\alt\state\in\dom\subd\hreg$, we have
\begin{equation}
\label{eq:3points}
\breg(\base,\alt\state)
	= \breg(\base,\state)
	+ \breg(\state,\alt\state)
	+ \braket{\nabla\hreg(\alt\state) - \nabla\hreg(\state)}{\state - \base}.
\end{equation}
\end{lemma}

\begin{proof}
By definition, we get:
\begin{equation}
\begin{aligned}
\breg(\base,\alt\state)
	&= \hreg(\base) - \hreg(\alt\state) - \braket{\nabla\hreg(\alt\state)}{\base - \alt\state}
	\\
\breg(\base,\state)\hphantom{'}
	&= \hreg(\base) - \hreg(\state) - \braket{\nabla\hreg(\state)}{\base - \state}
	\\
\breg(\state,\alt\state)
	&= \hreg(\state) - \hreg(\alt\state) - \braket{\nabla\hreg(\alt\state)}{\state - \alt\state}.
\end{aligned}
\end{equation}
The lemma then follows by adding the two last lines and subtracting the first.
\end{proof}

With all this at hand, we have the following upper and lower bounds:

\begin{proposition}
\label{prop:Bregman}
Let $\hreg$ be a $\strong$-strongly convex regularizer on $\feas$,
fix some $\base\in\feas$,
and let $\new\state = \prox_{\state}(\dstate)$ for $\state\in\dom\subd\hreg$, $\dstate\in\dspace$.
Then, we have:
\begin{subequations}
\begin{flalign}
\label{eq:Bregman-lower}
\breg(\base,\state)\hphantom{^{+}}
	&\geq \frac{\strong}{2} \norm{\state - \base}^{2}.
	\\
\label{eq:Bregman-old2new}
\breg(\base,\new\state)
	&\leq \breg(\base,\state)
	- \breg(\new\state,\state)
	+ \braket{\dstate}{\new\state - \base}
	\\
\label{eq:Bregman-old2new-alt}
	&\leq \breg(\base,\state)
	+ \braket{\dstate}{\state - \base}
	+ \frac{1}{2\strong} \dnorm{\dstate}^{2}
\end{flalign}
\end{subequations}
\end{proposition}

\begin{proof}[Proof of \eqref{eq:Bregman-lower}]
By the strong convexity of $\hreg$, we get
\begin{equation}
\hreg(\base)
	\geq \hreg(\state)
	+ \braket{\nabla\hreg(\state)}{\base - \state}
	+ \frac{\strong}{2} \norm{\base - \state}^{2}
\end{equation}
so \eqref{eq:Bregman-lower} follows by gathering all terms involving $\hreg$ and recalling the definition of $\breg(\base,\state)$.
\end{proof}

\begin{proof}[Proof of \labelcref{eq:Bregman-old2new,eq:Bregman-old2new-alt}]
By the three-point identity \eqref{eq:3points}, we readily obtain
\begin{equation}
\breg(\base,\state)
	= \breg(\base,\new\state)
	+ \breg(\new\state,\state)
	+ \braket{\nabla\hreg(\state) - \nabla\hreg(\new\state)}{\new\state - \base},
\end{equation}
and hence:

\begin{flalign}
\label{eq:upper-new2old}
\breg(\base,\new\state)
	&= \breg(\base,\state)
	- \breg(\new\state,\state)
	+ \braket{\nabla\hreg(\new\state) - \nabla\hreg(\state)}{\new\state - \base}
	\notag\\
	&\leq \breg(\base,\state)
	- \breg(\new\state,\state)
	+ \braket{\dstate}{\new\state - \base},
\end{flalign} where, in the last step, we used \eqref{eq:selection} and the fact that $\new\state = \mirror(\nabla\hreg(\state) + \dstate)$, by \eqref{eq:links-prox}, since $\new\state = \prox_{\state}(\dstate)$.
The above is just \eqref{eq:Bregman-old2new}, so the first part of our proof is complete.

To proceed with the proof of \eqref{eq:Bregman-old2new-alt}, note that \eqref{eq:upper-new2old} gives
\begin{flalign}
\breg(\base,\new\state)
	&\leq \breg(\base,\state)
	+ \braket{\dstate}{\state - \base}
	+ \braket{\dstate}{\new\state - \state}
	- \breg(\new\state,\state).
\end{flalign}
By Young's inequality \citep{Roc70}, we also have
\begin{equation}
\braket{\dstate}{\new\state - \state}
	\leq \frac{\strong}{2} \norm{\new\state - \state}^{2}
	+ \frac{1}{2\strong} \dnorm{\dstate}^{2},
\end{equation}
and hence
\begin{flalign}
\breg(\base,\new\state)
	&\leq \breg(\base,\state)
	+ \braket{\dstate}{\state - \base}
	+ \frac{1}{2\strong} \dnorm{\dstate}^{2}
	+ \frac{\strong}{2} \norm{\new\state - \state}^{2}
	- \breg(\new\state,\state)
	\notag\\
	&\leq \breg(\base,\state)
	+ \braket{\dstate}{\state - \base}
	+ \frac{1}{2\strong} \dnorm{\dstate}^{2},
\end{flalign}
with the last step following from \cref{lem:mirror} after plugging in $\state$ in place of $\base$.
\end{proof}

\section{Asymptotic convergence analysis}
\label{app:convergence}

Our goal in this appendix is to prove \cref{thm:convergence}.
Our proof strategy will be based on a two-pronged approach.
First, we will show that the pivot sequence $\act_{\run}$ satisfies a ``quasi-Fejér'' property \citep{Com01,CP15} with respect to the Bregman divergence.
This quasi-Fejér property allows us to show that the Bregman divergence $\breg(\eq,\act_{\run})$ with respect to a \acl{NE} $\eq$ of $\game$ converges.
To show that this limit is actually zero for \emph{some} \acl{NE}, we prove that, with probability $1$, the sequence $\act_{\run}$ admits a (random) subsequence that converges to a \acl{NE}.
The theorem then follows by combining these two results.

To carry all this out, we begin with an auxiliary lemma for the \ac{SPSA} estimation process of \cref{sec:feedback}:

\begin{lemma}
\label{lem:SPSA}
The \ac{SPSA} estimator $\est\payv = (\est\payv_{\play})_{\play\in\players}$ given by \eqref{eq:SPSA} satisfies
\begin{equation}
\exof{\est\payv_{\play}}
	= \del_{\play} \pay_{\play}^{\mix},
\end{equation}
with $\pay_{\play}^{\mix}$ as in \eqref{eq:smoothed}.
Moreover, we have $\norm{\del_{\play} \pay_{\play}^{\mix} - \del_{\play}\pay_{\play}}_{\infty} = \bigoh(\mix)$.
\end{lemma}

\begin{proof}
By the independence of the sampling directions $\unitvec_{\play}$, $\play\in\players$, we have
\begin{flalign}
\exof{\est\payv_{\play}}
	&= \frac{\vdim_{\play}/\mix}{\prod_{\playalt} \vol(\sphere_{\playalt})}
	\int_{\sphere_{1}} \dotsi \int_{\sphere_{\nPlayers}}
		\pay_{\play}(\state_{1} + \mix\unitvec_{1},\dotsc,\state_{\nPlayers} + \mix\unitvec_{\nPlayers})
		\unitvec_{\play}\,
		\dd \unitvec_{1}\dotsm \dd \unitvec_{\nPlayers}
	\notag\\
	&= \frac{\vdim_{\play}/\mix}{\prod_{\playalt} \vol(\mix\sphere_{\playalt})}
	\int_{\mix\sphere_{1}} \dotsi \int_{\mix\sphere_{\nPlayers}}
		\pay_{\play}(\state_{1} + \unitvec_{1},\dotsc,\state_{\nPlayers} + \unitvec_{\nPlayers})
		\frac{\unitvec_{\play}}{\norm{\unitvec_{\play}}}\,
		\dd \unitvec_{1}\dotsm \dd \unitvec_{\nPlayers}
	\notag\\
	&= \frac{\vdim_{\play}/\mix}{\prod_{\playalt} \vol(\mix\sphere_{\playalt})}
	\int_{\mix\sphere_{\play}} \int_{\prod_{\playalt\neq\play} \mix\sphere_{\playalt}}
		\pay_{\play}(\state_{\play} + \unitvec_{\play};\state_{-\play} + \unitvec_{-\play})
		\frac{\unitvec_{\play}}{\norm{\unitvec_{\play}}}\,
		\dd \unitvec_{\play} \dd \unitvec_{-\play}
	\notag\\
	&= \frac{\vdim_{\play}/\mix}{\prod_{\playalt} \vol(\mix\sphere_{\playalt})}
	\int_{\mix\ball_{\play}} \int_{\prod_{\playalt\neq\play} \mix\sphere_{\playalt}}
		\del_{\play}\pay_{\play}(\state_{\play} + \pert_{\play};\state_{-\play} + \unitvec_{-\play})
		\dd \pert_{\play} \dd \unitvec_{-\play},
\end{flalign}
where, in the last line, we used the identity
\begin{equation}
\label{eq:Stokes}
\nabla\int_{\mix\ball} \obj(\state+\pert) \dd\pert
	= \int_{\mix\sphere} \obj(\state+\unitvec) \frac{\unitvec}{\norm{\unitvec}} \dd\unitvec
\end{equation}
which, in turn, follows from Stokes' theorem \citep{Lee03,FKM05}.
Since $\vol(\mix\ball_{\play}) = (\mix/\vdim_{\play}) \vol(\mix\sphere_{\play})$, the above yields $\exof{\est\payv_{\play}}= \del_{\play} \pay_{\play}^{\mix}$ with $\pay_{\play}^{\mix}$ given by \eqref{eq:smoothed}.

For the second part of the lemma, let $\Lip_{\play}$ denote the Lipschitz constant of $\payv_{\play}$, \ie $\dnorm{\payv_{\play}(\alt\state) - \payv_{\play}(\state)} \leq \Lip_{\play} \norm{\alt\state - \state}$ for all $\state,\alt\state\in\feas$.
Then, for all $\pert_{\play}\in\mix\ball_{\play}$ and all $\unitvec_{\playalt} \in \mix\sphere_{\playalt}$, $\playalt\neq\play$, we have
\begin{equation}
\norm{
	\del_{\play} \pay_{\play}(\state_{\play} + \pert_{\play}; \state_{-\play} + \unitvec_{-\play})
	- \del_{\play} \pay_{\play}(\state)}
	\leq \Lip_{\play} \sqrt{\norm{\pert_{\play}}^{2} + \insum_{\playalt\neq\play} \norm{\unitvec_{\playalt}}^{2}}
	\leq \Lip_{\play} \sqrt{\nPlayers} \mix.
\end{equation}
Our assertion then follows by integrating and differentiating under the integral sign. 
\end{proof}

With this basic estimate at hand, we proceed to establish the convergence of the Bregman divergence relative to the game's \aclp{NE}:

\begin{proposition}
\label{prop:quasiFejer}
Let $\eq$ be a \acl{NE} of $\game$.
Then, with assumptions as in \cref{thm:convergence}, the Bregman divergence $\breg(\eq,\act_{\run})$ converges \as to a finite random variable $\breg_{\infty}$.
\end{proposition}

\begin{remark*}
For expository reasons, we tacitly assume above (and in what follows) that $\game$ satisfies \eqref{eq:DSC} with weights $\lambda_{\play} = 1$ for all $\play\in\players$.
If this is not the case, the Bregman divergence $\breg(\base,\state)$ should be replaced by the weight-adjusted variant
\begin{equation}
\label{eq:Bregman-lambda}
\breg^{\lambda}(\base,\state)
	= \sum_{\play\in\players} \lambda_{\play} \breg(\base_{\play},\state_{\play}).
\end{equation}
Since this adjustment would force us to carry around all player indices, the presentation would become significantly more cumbersome;
to avoid this, we stick with the simpler, unweighted case.
\end{remark*}

\begin{proof}
Let $\breg_{\run} = \breg(\eq,\act_{\run})$ for some \acl{NE} $\eq$ of $\game$ and write
\begin{equation}
\label{eq:feedback}
\est\payv_{\run}
	= \payv(\act_{\run})
	+ \noise_{\run+1}
	+ \bias_{\run},
\end{equation}
where, recalling the setup of \cref{sec:feedback} in the main body of the paper,
the noise process $\noise_{\run+1} = \est\payv_{\run} - \exof{\est\payv_{\run} \given \filter_{\run}}$ is an $\filter_{\run}$-adapted \acl{MDS}
and
$\bias_{\run} = \payv^{\mix_{\run}}(\act_{\run}^{\mix_{\run}}) - \payv(\act_{\run})$ denotes the systematic bias of the estimator $\est\payv_{\run}$.%
\footnote{Recall here that $\act_{\play}^{\mix}$, $\play\in\players$, denotes the $\mix$-adjusted pivot $\act_{\play}^{\mix} = \act_{\play} + \radius_{\play}^{-1} \mix (\act_{\play} - \base_{\play})$, \ie including the feasibility adjustment $\radius_{\play}^{-1} (\act_{\play} - \base_{\play})$.}
Then, by \cref{prop:Bregman}, we have
\begin{flalign}
\label{eq:Bregman-new-bound}
\breg_{\run+1}
	= \breg(\eq,\prox_{\act_{\run}}(\step_{\run}\est\payv_{\run}))
	&\leq \breg(\eq,\act_{\run})
	+ \step_{\run} \braket{\est\payv_{\run}}{\act_{\run} - \eq}
	+ \frac{\step_{\run}^{2}}{2\strong} \dnorm{\est\payv_{\run}}^{2}
	\notag\\
	&= \breg_{\run}
	+ \step_{\run} \braket{\payv(\act_{\run}) + \noise_{\run+1} + \bias_{\run}}{\act_{\run} - \eq}
	+ \frac{\step_{\run}^{2}}{2\strong} \dnorm{\est\payv_{\run}}^{2}
	\notag\\
	&\leq \breg_{\run}
	+ \step_{\run} \snoise_{\run+1}
	+ \step_{\run} \sbias_{\run}
	+ \frac{\step_{\run}^{2}}{2\strong} \dnorm{\est\payv_{\run}}^{2},
\end{flalign}
where, in the last line, we set
$\snoise_{\run+1} = \braket{\noise_{\run+1}}{\act_{\run} - \eq}$,
$\sbias_{\run} = \braket{\bias_{\run}}{\act_{\run} - \eq}$,
and we used the variational characterization \eqref{eq:VI} of \aclp{NE} of monotone games.
Thus, conditioning on $\filter_{\run}$ and taking expectations, we get
\begin{flalign}
\exof{\breg_{\run+1} \given \filter_{\run}}
	&\leq \breg_{\run}
	+ \exof{\snoise_{\run+1} \given \filter_{\run}}
	+ \step_{\run} \exof{\sbias_{\run} \given \filter_{\run}}
	+ \frac{\step_{\run}^{2}}{2\strong} \exof{\dnorm{\est\payv_{\run}}^{2} \given \filter_{\run}}
	\notag\\
	&\leq \breg_{\run}
	+ \step_{\run} \exof{\sbias_{\run} \given \filter_{\run}}
	+ \frac{\vbound^{2}}{2\strong} \frac{\step_{\run}^{2}}{\mix_{\run}^{2}}.
\end{flalign}
where
we set $\vbound^{2} = \sum_{\play} \vdim_{\play}^{2} \max_{\state\in\feas} \abs{\pay_{\play}(\state)}^{2}$
and
we used the fact that $\act_{\run}$ is $\filter_{\run}$-measurable, so
\begin{equation}
\exof{\snoise_{\run+1} \given \filter_{\run}}
	= \braket{\exof{\noise_{\run+1}\given\filter_{\run}}}{\act_{\run} - \eq}
	= 0.
\end{equation}
Finally, by \cref{lem:SPSA}, we have
\begin{flalign}
\dnorm{\bias_{\run}}
	&= \dnorm{\payv^{\mix_{\run}}(\act_{\run}^{\mix_{\run}}) - \payv(\act_{\run})}
	\notag\\
	&\leq \dnorm{\payv^{\mix_{\run}}(\act_{\run}^{\mix_{\run}}) - \payv(\act_{\run}^{\mix_{\run}})}
	+ \dnorm{\payv(\act_{\run}^{\mix_{\run}}) - \payv(\act_{\run})}
	\notag\\
	&= \bigoh(\mix_{\run}),
\end{flalign}
where we used the fact that $\payv$ is Lipschitz continuous and $\norm{\payv^{\mix} - \payv}_{\infty} = \bigoh(\mix)$.
This shows that there exists some $\bbound > 0$ such that $\sbias_{\run} \leq \bbound \mix_{\run}$;
as a consequence, we obtain
\begin{equation}
\exof{\breg_{\run+1} \given \filter_{\run}}
	\leq \breg_{\run}
	+ \bbound \step_{\run} \mix_{\run}
	+ \frac{\vbound^{2}}{2\strong} \frac{\step_{\run}^{2}}{\mix_{\run}^{2}}.
\end{equation}

Now, letting $R_{\run} = \breg_{\run} + \sum_{\runalt=\run}^{\infty} \bracks{\bbound \step_{\runalt} \mix_{\runalt} + (2\strong)^{-1} \vbound^{2} \step_{\runalt}^{2}/\mix_{\runalt}^{2} }$, the estimate \eqref{eq:Bregman-new-bound} gives
\begin{flalign}
\exof{R_{\run+1} \given \filter_{\run}}
	&= \exof{\breg_{\run+1} \given \filter_{\run}}
	+ \sum_{\runalt=\run+1}^{\infty} \bracks*{\bbound\step_{\runalt}\mix_{\runalt} + \frac{\vbound^{2}}{2\strong} \frac{\step_{\runalt}^{2}}{\mix_{\runalt}^{2}}}
	\notag\\
	&\leq \breg_{\run}
	+ \bbound \step_{\run} \mix_{\run}
	+ \frac{\vbound^{2}}{2\strong} \frac{\step_{\run}^{2}}{\mix_{\run}^{2}}
	+ \sum_{\runalt=\run+1}^{\infty} \bracks*{\bbound\step_{\runalt}\mix_{\runalt} + \frac{\vbound^{2}}{2\strong} \frac{\step_{\runalt}^{2}}{\mix_{\runalt}^{2}}}
	\notag\\
	&\leq \breg_{\run}
	+ \sum_{\runalt=\run}^{\infty} \bracks*{\bbound\step_{\runalt}\mix_{\runalt} + \frac{\vbound^{2}}{2\strong} \frac{\step_{\runalt}^{2}}{\mix_{\runalt}^{2}}}
	\notag\\
	&= R_{\run},
\end{flalign}
\ie $R_{\run}$ is an $\filter_{\run}$-adapted supermartingale.%
\footnote{In particular, this shows that $\exof{\breg_{\run} \given \filter_{\run-1}}$ is quasi-Fejér in the sense of \cite{Com01}.}
Since the series $\sum_{\run=\start}^{\infty} \step_{\run} \mix_{\run}$ and $\sum_{\run=\start}^{\infty} \step_{\run}^{2} / \mix_{\run}^{2}$ are both summable, it follows that
\begin{flalign}
\exof{R_{\run}}
	&= \exof{\exof{R_{\run} \given \filter_{\run-1}}}
	\notag\\
	&\leq \exof{R_{\run-1}}
	\leq \dotsm
	\leq \exof{R_{\start}}
	\notag\\
	&\leq \exof{\breg_{\start}}
	+ \sum_{\run=\start}^{\infty}
		\bracks*{\bbound\step_{\run}\mix_{\run}
			+ \frac{\vbound^{2}}{2\strong}\frac{\step_{\run}^{2}}{\mix_{\run}^{2}}}
	\notag\\
	&<\infty
\end{flalign}
\ie $R_{\run}$ is uniformly bounded in $L^{1}$.
Thus, by Doob's convergence theorem for supermartingales \citep[Theorem~2.5]{HH80}, it follows that $R_{\run}$ converges \as to some finite random variable $R_{\infty}$.
In turn, by inverting the definition of $R_{\run}$, it follows that $\breg_{\run}$ converges \as to some random variable $\breg_{\infty}$, as claimed.
\end{proof}

\begin{proposition}
\label{prop:subsequence}
Suppose that the assumptions of \cref{thm:convergence} hold.
Then, with probability $1$, there exists a \textpar{random} subsequence $\act_{\run_{\runalt}}$ of \eqref{eq:MD-0} which converges to \acl{NE}.
\end{proposition}

\begin{proof}
We begin with the technical observation that the set $\eqset$ of \aclp{NE} of $\game$ is closed (and hence, compact).
Indeed, let $\eq_{\run}$, $\run=1,2,\dotsc$, be a sequence of \aclp{NE} converging to some limit point $\eq\in\feas$;
to show that $\eqset$ is closed, it suffices to show that $\eq\in\feas$.
However, since \aclp{NE} of $\game$ satisfy the variational characterization \eqref{eq:VI}, we also have $\braket{\payv(\state)}{\state - \eq_{\run}} \leq 0$ for all $\state\in\feas$.
Hence, with $\eq_{\run}\to\eq$ as $\run\to\infty$, it follows that
\begin{equation}
\braket{\payv(\state)}{\state-\eq}
	= \lim_{\run\to\infty} \braket{\payv(\state)}{\state - \eq_{\run}}
	\leq 0
	\quad
	\text{for all $\state\in\feas$},
\end{equation}
\ie $\eq$ satisfies \eqref{eq:VI}.
Since $\game$ is monotone, we conclude that $\eq$ is a \acl{NE}, as claimed.

Suppose now ad absurdum that, with positive probability, the pivot sequence $\act_{\run}$ generated by \eqref{eq:MD-0} admits no limit points in $\eqset$.%
\footnote{We assume here without loss of generality that $\eqset\neq\feas$; otherwise, there is nothing to show.}
Conditioning on this event, and given that $\eqset$ is compact, there exists a (nonempty) compact set $\cpt\subset\feas$ such that $\cpt\cap\eqset = \varnothing$ and $\act_{\run}\in\cpt$ for all sufficiently large $\run$.
Moreover, by \eqref{eq:VI}, we have $\braket{\payv(\state)}{\state-\eq} < 0$ whenever $\state\in\cpt$ and $\eq\in\eqset$.
Therefore, by the continuity of $\payv$ and the compactness of $\eqset$ and $\cpt$, there exists some $\farbound>0$ such that
\begin{equation}
\label{eq:farbound}
\braket{\payv(\state)}{\state - \eq}	
	\leq -\farbound
	\quad
	\text{for all $\state\in\cpt$, $\eq\in\feas$}.
\end{equation}

To proceed, fix some $\eq\in\eqset$ and let $\breg_{\run} = \breg(\eq,\act_{\run})$ as in the proof of \cref{prop:quasiFejer}.
Then, telescoping \eqref{eq:Bregman-new-bound} yields the estimate
\begin{flalign}
\label{eq:Bregman-run}
\breg_{\run+1}
	&\leq \breg_{\start}
	+ \sum_{\runalt=\start}^{\run} \step_{\runalt} \braket{\payv(\act_{\run})}{\act_{\run} - \eq}
	+ \sum_{\runalt=\start}^{\run} \step_{\runalt} \snoise_{\runalt+1}
	+ \sum_{\runalt=\start}^{\run} \step_{\runalt} \sbias_{\runalt}
	+ \sum_{\runalt=\start}^{\run} \frac{\step_{\runalt}^{2}}{2\strong} \dnorm{\est\payv_{\run}}^{2},
\end{flalign}
where, as in the proof of \cref{prop:quasiFejer}, we set
\begin{flalign}
\label{eq:noise-scalar}
\snoise_{\run+1}
	&= \braket{\noise_{\run+1}}{\act_{\run} - \eq}
\intertext{and}
\label{eq:bias-scalar}
\sbias_{\run}
	&= \braket{\bias_{\run}}{\act_{\run} - \eq}.
\end{flalign}
Subsequently, letting $\tau_{\run} = \sum_{\runalt=\start}^{\run} \step_{\runalt}$ and using \eqref{eq:farbound}, we obtain
\begin{equation}
\label{eq:Bregman-run-bound}
\breg_{\run+1}
	\leq \breg_{\start}
	- \tau_{\run} \bracks*{
		\farbound
		- \frac{\sum_{\runalt=\start}^{\run} \step_{\runalt} \snoise_{\runalt+1}}{\tau_{\run}}
		- \frac{\sum_{\runalt=\start}^{\run} \step_{\runalt} \sbias_{\runalt}}{\tau_{\run}}
		- \frac{(2\strong)^{-1} \sum_{\runalt=\start}^{\run} \step_{\runalt}^{2} \dnorm{\est\payv_{\runalt}}^{2}}{\tau_{\run}}
	}.
\end{equation}

Since $\noise_{\run}$ is a \acl{MDS} with respect to $\filter_{\run}$, we have $\exof{\snoise_{\run+1} \given \filter_{\run}} = 0$ (recall that $\act_{\run}$ is $\filter_{\run}$-measurable by construction).
Moreover, by construction, there exists some constant $\noisedev>0$ such that
\begin{equation}
\dnorm{\noise_{\run+1}}^{2}
	\leq \frac{\noisevar}{\mix_{\run}^{2}},
\end{equation}
and hence:
\begin{flalign}
\sum_{\run=\start}^{\infty} \step_{\run}^{2} \exof{\snoise_{\run+1}^{2} \given \filter_{\run}}
	&\leq \sum_{\run=\start}^{\infty} \step_{\run}^{2} \norm{\act_{\run} - \eq}^{2} \exof{\dnorm{\noise_{\run+1}}^{2}\given\filter_{\run}}
	\notag\\
	&\leq \diam(\feas)^{2} \noisevar \sum_{\run=\start}^{\infty} \frac{\step_{\run}^{2}}{\mix_{\run}^{2}}
	< \infty.
\end{flalign}
Therefore, by the law of large numbers for \aclp{MDS} \citep[Theorem~2.18]{HH80}, we conclude that $\tau_{\run}^{-1} \sum_{\runalt=\start}^{\run} \step_{\runalt} \snoise_{\runalt+1}$ converges to $0$ with probability $1$.

For the third term in the brackets of \eqref{eq:Bregman-run-bound} we have $\sbias_{\run}\to0$ as $\run\to\infty$ \as.
Since $\sum_{\run=\start}^{\infty} \step_{\run} = \infty$, it follows $\sum_{\runalt=\start}^{\run} \step_{\runalt} \sbias_{\runalt} \big/ \sum_{\runalt=\start}^{\run} \step_{\runalt} \to 0$.

Finally, for the last term in the brackets of \eqref{eq:Bregman-run-bound}, let $S_{\run+1} = \sum_{\runalt=\start}^{\run} \step_{\runalt}^{2} \dnorm{\est\payv_{\runalt}}^{2}$.
Since $\est\payv_{\runalt}$ is $\filter_{\run}$-measurable for all $\runalt=\running,\run-1$, we have
\begin{equation}
\exof{S_{\run+1} \given \filter_{\run}}
	= \exof*{%
		\sum_{\runalt=\start}^{\run-1} \step_{\runalt}^{2} \dnorm{\est\payv_{\runalt}}^{2}
		+ \step_{\run}^{2} \dnorm{\est\payv_{\run}}^{2} \given \filter_{\run}
		}
	= S_{\run} + \step_{\run}^{2} \exof{\dnorm{\est\payv_{\run}}^{2} \given \filter_{\run}}
	\geq S_{\run},
\end{equation}
\ie $S_{\run}$ is a submartingale with respect to $\filter_{\run}$.
Furthermore, by the law of total expectation, we also have
\begin{equation}
\exof{S_{\run+1}}
	= \exof{\exof{S_{\run+1} \given \filter_{\run}}}
	\leq \vbound^{2} \sum_{\runalt=\start}^{\run} \frac{\step_{\runalt}^{2}}{\mix_{\runalt}^{2}}
	\leq \vbound^{2} \sum_{\runalt=\start}^{\infty} \frac{\step_{\runalt}^{2}}{\mix_{\runalt}^{2}}
	< \infty,
\end{equation}
implying in turn that $S_{\run}$ is uniformly bounded in $L^{1}$.
Hence, by Doob's submartingale convergence theorem \citep[Theorem~2.5]{HH80}, we conclude that $S_{\run}$ converges to some (almost surely finite) random variable $S_{\infty}$ with $\exof{S_{\infty}} < \infty$.
Consequently, we have $\lim_{\run\to\infty} S_{\run+1}/\tau_{\run} = 0$ with probability $1$.

Applying all of the above to the estimate \eqref{eq:Bregman-run-bound}, we get $\breg_{\run+1} \leq \breg_{\start} - \farbound \tau_{\run} / 2$ for sufficiently large $\run$,
and hence, $\breg(\eq,\act_{\run}) \to -\infty$, a contradiction.
Going back to our original assumption, this shows that at least one of the limit points of $\act_{\run}$ must lie in $\eqset$, so our proof is complete.
\end{proof}

We are finally in a position to prove \cref{thm:convergence} regarding the convergence of \eqref{eq:MD-0}:

\begin{proof}[Proof of \cref{thm:convergence}]
By \cref{prop:subsequence}, there exists a (possibly random) \acl{NE} $\eq$ of $\game$ such that $\norm{\act_{\run_{\runalt}} - \eq} \to 0$ for some (random) subsequence $\act_{\run_{\runalt}}$.
By the assumed reciprocity of the Bregman divergence, this implies that $\liminf_{\run\to\infty} \breg(\eq,\act_{\run}) = 0$ \as.
Since $\lim_{\run\to\infty} \breg(\eq,\act_{\run})$ exists with probability $1$ (by \cref{prop:quasiFejer}), it follows that
\begin{equation}
\lim_{\run\to\infty} \breg(\eq,\act_{\run})
	= \liminf_{\run\to\infty} \breg(\eq,\act_{\run})
	= 0,
\end{equation}
\ie $\act_{\run}$ converges to $\eq$ by the first part of \cref{prop:Bregman}.
Since $\mix_{\run}\to0$ and $\norm{\est\act_{\run} - \act_{\run}} = \mix_{\run} \norm{\pertvar_{\run}}= \bigoh(\mix_{\run})$, our claim follows.
\end{proof}

\section{Finite-time analysis and rates of convergence}
\label{app:rates}

We now turn to the finite-time analysis of \eqref{eq:MD-0}.
To begin, we briefly recall that a game $\game$ is \emph{$\str$-strongly monotone} if it satisfies the condition
\begin{equation}
\label{eq:DSC-str}
\tag{$\str$-DSC}
\sum_{\play\in\players} \lambda_{\play} \braket{\payv_{\play}(\alt \state) - \payv_{\play}(\state)}{\alt \state_{\play} - \state_{\play}}
	\leq -\frac{\str}{2} \norm{\state - \alt\state}^{2}
\end{equation}
for some $\lambda_{\play},\str>0$ and for all $\state,\alt \state\in\feas$.
Our aim in what follows will be to prove the following convergence rate estimate for multi-agent \acl{MD} in strongly monotone games:

\begin{theorem}
\label{thm:rates}
Let $\eq$ be the \textpar{unique} \acl{NE} of a $\str$-strongly monotone game.
Then:
\begin{enumerate}
[\hspace{2em}\itshape a\upshape)]
\item
If the players have access to a gradient oracle satisfying \eqref{eq:oracle} and they follow \eqref{eq:MD} with Euclidean projections and step-size sequence $\step_{\run} = \step/\run$ for some $\step > 1/\str$, we have
\begin{equation}
\label{eq:rate-oracle}
\exof{\norm{\act_{\run} - \eq}^{2}}
	= \bigoh(\run^{-1}).
\end{equation}
\item
If the players only have bandit feedback and they follow \eqref{eq:MD-0} with Euclidean projections and parameters $\step_{\run} = \step/\run$ and $\mix_{\run} = \mix/\run^{1/3}$ with $\step > 1/(3\str)$ and $\mix>0$, we have
\begin{equation}
\label{eq:rate}
\exof{\norm{\est\act_{\run} - \eq}^{2}}
	= \bigoh(\run^{-1/3}).
\end{equation}
\end{enumerate}
\end{theorem}

\begin{remark*}
\cref{thm:rate} is recovered by the second part of \cref{thm:rates} above;
the first part (which was alluded to in the main paper) serves as a benchmark to quantify the gap between bandit and oracle feedback.
\end{remark*}

For the proof of \cref{thm:rates} we will need the following lemma on numerical sequences, a version of which is often attributed to \cite{Chu54}:

\begin{lemma}
\label{lem:Chung}
Let $\seq_{\run}$, $\run=\running$, be a non-negative sequence such that
\begin{equation}
\label{eq:seq-bound}
\seq_{\run+1}
	\leq \seq_{\run} \parens*{1 - \frac{\pbound}{\run^{\pexp}}}
	+ \frac{\qbound}{\run^{\pexp+\qexp}}
\end{equation}
where $0<\pexp\leq1$, $\qexp>0$, and $\pbound, \qbound>0$.
Then, assuming $\pbound>\qexp$ if $\pexp=1$, we have
\begin{equation}
\label{eq:seq-rate}
\seq_{\run}
	\leq \frac{\qbound}{\rbound} \frac{1}{\run^{\qexp}}
	+ o\parens*{\frac{1}{\run^{\qexp}}},
\end{equation}
with $\rbound = \pbound$ if $\pexp<1$ and $\rbound = \pbound - \qexp$ if $\pexp=1$.
\end{lemma}

\begin{proof}
Clearly, it suffices to show that $\limsup_{\run\to\infty} \run^{\qexp} \seq_{\run} \leq \qbound/\rbound$.
To that end, write $\qexp_{\run} = \run \bracks{(1+1/\run)^{\qexp} - 1}$, so $(1+1/\run)^{\qexp} = 1 + \qexp_{\run}/\run$ and $\qexp_{\run} \to \qexp$ as $\run\to\infty$.
Then, multiplying both sides of \eqref{eq:seq-bound} by $(\run+1)^{\qexp}$ and letting $\var\seq_{\run} = \seq_{\run} \run^{\qexp}$, we get
\begin{flalign}
\label{eq:seq-rate-temp1}
\var\seq_{\run+1}
	&\leq \seq_{\run}(\run+1)^{\qexp} \parens*{1 - \frac{\pbound}{\run^{\pexp}}}
	+ \frac{\qbound (\run+1)^{\qexp}}{\run^{\pexp+\qexp}}
	\notag\\
	&= \var\seq_{\run} \parens*{1 + \frac{\qexp_{\run}}{\run}} \parens*{1 - \frac{\pbound}{\run^{\pexp}}}
	+ \frac{\qbound(1+\qexp_{\run}/\run)}{\run^{\pexp}}
	\notag\\
	&= \var\seq_{\run} \bracks*{1 + \frac{\qexp_{\run}}{\run} - \frac{\pbound}{\run^{\pexp}} + \bigoh\parens*{\frac{1}{\run^{\pexp+1}}}}
	+ \frac{\qbound_{\run}}{\run^{\pexp}},
\end{flalign}
where we set $\qbound_{\run} = \qbound(1+\qexp_{\run}/\run)$, so $\qbound_{\run}\to \qbound$ as $\run\to\infty$.
Then, under the assumption that $\pbound>\qexp$ when $\pexp=1$, \eqref{eq:seq-rate-temp1} can be rewritten as
\begin{equation}
\label{eq:seq-rate-temp2}
\var\seq_{\run+1}
	\leq \var\seq_{\run} \parens*{1 - \frac{\rbound_{\run}}{\run^{\pexp}}}
	+ \frac{\qbound_{\run}}{\run^{\pexp}},
\end{equation}
for some sequence $\rbound_{\run}$ with $\rbound_{\run}\to \rbound$ as $\run\to\infty$.

Now, fix some small enough $\eps>0$.
From \eqref{eq:seq-rate-temp2}, we readily get
\begin{equation}
\label{eq:seq-rate-temp3}
\var\seq_{\run+1}
	\leq \var\seq_{\run}
	- \frac{\rbound_{\run} \var\seq_{\run} - \qbound_{\run}}{\run^{\pexp}}.
\end{equation}
Since $\rbound_{\run}\to\rbound$ and $\qbound_{\run}\to\qbound$ as $\run\to\infty$, we will have $\rbound_{\run} > \rbound-\eps$ and $\qbound_{\run} < \qbound+\eps$ for all $\run$ greater than some $\run_{\eps}$.
Thus, if $\run\geq\run_{\eps}$ and $(\rbound-\eps) \var\seq_{\run} - (\qbound+\eps) > \eps$, we will also have
\begin{equation}
\var\seq_{\run+1}
	\leq \var\seq_{\run} - \frac{\rbound_{\run}\var\seq_{\run} - \qbound_{\run}}{\run^{\pexp}}
	\leq \var\seq_{\run} - \frac{(\rbound - \eps)\var\seq_{\run} - (\qbound+\eps)}{\run^{\pexp}}
	\leq \var\seq_{\run} - \frac{\eps}{\run^{\pexp}}.
\end{equation}
The above shows that, as long as $\var\seq_{\run} > (\qbound+2\eps) / (\rbound - \eps)$, $\var\seq_{\run}$ will decrease at least by $\eps/\run^{\pexp}$ at each step.
In turn, since $\sum_{\run=\start}^{\infty} (1/\run^{\pexp}) = \infty$, it follows by telescoping that $\limsup_{\run\to\infty} \var\seq_{\run} \leq (\qbound+2\eps)/(\rbound - \eps)$.
Hence, with $\eps$ arbitrary, we conclude that $\limsup_{\run\to\infty} \seq_{\run} \run^{\qexp} \leq \qbound/\rbound$, as claimed.
\end{proof}

\begin{proof}[Proof of \cref{thm:rates}]
We begin with the second part of the theorem;
the first part will follow by setting some estimates equal to zero, so the analysis is more streamlined that way.
Also, as in the previous section, we tacitly assume that \eqref{eq:DSC-str} holds with weights $\lambda_{\play} = 1$ for all $\play\in\players$.
If this is not the case, the Bregman divergence $\breg(\base,\state)$ should be replaced by the weight-adjusted variant \eqref{eq:Bregman-lambda}, but this would only make the presentation more difficult to follow, so we omit the details.

The main component of our proof is the estimate \eqref{eq:Bregman-new-bound}, which, for convenience (and with notation as in the previous section), we also reproduce below:
\begin{flalign}
\label{eq:Bregman-new-bound2}
\breg_{\run+1}
	\leq \breg_{\run}
	+ \step_{\run} \braket{\payv(\act_{\run})}{\act_{\run} - \eq}
	+ \step_{\run} \snoise_{\run+1}
	+ \step_{\run} \sbias_{\run}
	+ \frac{\step_{\run}^{2}}{2\strong} \dnorm{\est\payv_{\run}}^{2}.
\end{flalign}
In the above, since the algorithm is run with Euclidean projections, $\breg_{\run} = \frac{1}{2} \norm{\act_{\run} - \sol}^{2}$;
other than that, $\snoise_{\run}$ and $\sbias_{\run}$ are defined as in \eqref{eq:noise-scalar} and \eqref{eq:bias-scalar} respectively.
Since the game is $\str$-strongly monotone and $\eq$ is a \acl{NE}, we further have
\begin{equation}
\braket{\payv(\act_{\run})}{\act_{\run} - \eq}
	\leq \braket{\payv(\act_{\run}) - \payv(\eq)}{\act_{\run} - \eq}
	\leq - \frac{\str}{2} \norm{\act_{\run} - \eq}^{2}
	= - \str \breg_{\run},
\end{equation}
so \eqref{eq:Bregman-new-bound2} becomes
\begin{equation}
\label{eq:Bregman-new-strong}
\breg_{\run+1}
	\leq (1 - \str\step_{\run}) \breg_{\run}
	+ \step_{\run} \snoise_{\run+1}
	+ \step_{\run} \sbias_{\run}
	+ \frac{\step_{\run}^{2}}{2\strong} \dnorm{\est\payv_{\run}}^{2}.
\end{equation}
Thus, letting $\bar\breg_{\run} = \exof{\breg_{\run}}$ and taking expectations, we obtain
\begin{equation}
\label{eq:Bregman-new-strong-mean}
\bar\breg_{\run+1}
	\leq (1 - \str\step_{\run}) \bar\breg_{\run}
	+ \bbound \step_{\run} \mix_{\run}
	+ \frac{\vbound^{2}}{2\strong} \frac{\step_{\run}^{2}}{\mix_{\run}^{2}},
\end{equation}
with $\bbound$ and $\vbound$ defined as in the proof of \cref{thm:convergence} in the previous section.

Now, substituting $\step_{\run} = \step/\run^{\pexp}$ and $\mix_{\run} = \mix/\run^{\qexp}$ in \eqref{eq:Bregman-new-strong-mean} readily yields
\begin{equation}
\label{eq:Bregman-new-steps}
\bar\breg_{\run+1}
	\leq \parens*{1 - \frac{\str\step}{\run^{\pexp}}} \bar\breg_{\run}
	+ \frac{\bbound\step\mix}{\run^{\pexp+\qexp}}
	+ \frac{\vbound^{2}\step^{2}\mix^{2}}{2\strong \run^{2(\pexp - \qexp)}}.
\end{equation}
Hence, taking $\pexp=1$ and $\qexp=1/3$, the last two exponents are equated, leading to the estimate
\begin{equation}
\bar\breg_{\run+1}
	\leq \parens*{1 - \frac{\str\step}{\run}} \bar\breg_{\run}
	+ \frac{\cbound}{\run^{4/3}},
\end{equation}
with $\cbound = \step\mix\bbound + (2\strong)^{-1} \step^{2}\mix^{2}\vbound^{2}$.
Thus, with $\str\step>1/3$, applying \cref{lem:Chung} with $\pexp=1$ and $\qexp = 1/3$, we finally obtain $\bar\breg_{\run} = \bigoh(1/\run^{1/3})$.

The proof for the oracle case is similar:
the key observation is that the bound \eqref{eq:Bregman-new-strong-mean} becomes
\begin{equation}
\label{eq:Bregman-new-strong-oracle}
\bar\breg_{\run+1}
	\leq (1 - \str\step_{\run}) \bar\breg_{\run}
	+ \frac{\vbound^{2}}{2\strong} \step_{\run}^{2},
\end{equation}
with $\vbound$ defined as in \eqref{eq:oracle}.
Hence, taking $\step_{\run} = \step/\run$ with $\str\step>1$ and applying again \cref{lem:Chung} with $\pexp=\qexp=1$, we obtain $\bar\breg_{\run} = \bigoh(1/\run)$ and our proof is complete.
\end{proof}

To conclude, we note that the $\bigoh(1/\run^{1/3})$ bound of \cref{thm:rates} cannot be readily improved by choosing a different step-size schedule of the form $\step_{\run}\propto 1/\run^{\pexp}$ for some $\pexp<1$.
Indeed, applying \cref{lem:Chung} to the estimate \eqref{eq:Bregman-new-steps} yields a bound which is either $\bigoh(1/\run^{\qexp})$ or $\bigoh(1/\run^{\pexp - 2\qexp})$, depending on which exponent is larger.
Equating the two exponents (otherwise, one term would be slower than the other), we get $\qexp = \pexp/3$, leading again to a $\bigoh(1/\run^{1/3})$ bound.
Unless one has finer control on the bias/variance of the \ac{SPSA} gradient estimator used in \eqref{eq:MD-0}, we do not see a way of improving this bound in the current context.

\bibliographystyle{ormsv080}
\bibliography{IEEEabrv,../bibtex/Bibliography-MD}

\end{document}